\documentclass[aps,
longbibliography,
superscriptaddress]{revtex4-1}
\usepackage{amsmath,amsthm,amssymb,graphicx}
\usepackage{bm,mathrsfs,dsfont,nicefrac}
\usepackage{braket,comment}
\usepackage{xcolor}
\usepackage{float}
\theoremstyle{remark}
\usepackage{algorithm, algpseudocode}
\usepackage{algorithmicx}
\usepackage{stmaryrd} %For box command
\usepackage{youngtab}
\usepackage{ytableau}
\algdef{SE}[VARIABLES]{Variables}{EndVariables}
{\algorithmicvariables}
{\algorithmicend\algorithmicvariables}
\algnewcommand{\algorithmicvariables}{\textbf{Global variables}}
\newtheorem{problem}{Problem}
\newtheorem{definition}{Definition}
\newtheorem{theorem}{Theorem}

\newtheorem{lemma}[theorem]{Lemma}

\newtheorem{remark}{Remark}

\newtheorem{kase}{Case}
\algnewcommand{\algorithmicgoto}{\textbf{go to}}%
\algnewcommand{\Goto}[1]{\algorithmicgoto~\ref{#1}}%

\begin{document}
\title{
%    Algorithmic approach to derive 
    Group-covariant extreme and quasi-extreme channels}
\author{Laleh Memarzadeh}
 \email{memarzadeh@sharif.edu}
\affiliation{%
	Department of Physics, Sharif University of Technology, Tehran 11365-9161, Iran}
\author{Barry C.\ Sanders}%
 \email{sandersb@ucalgary.ca}
 \homepage{http://iqst.ca/people/peoplepage.php?id=4}
\affiliation{%
	Institute for Quantum Science and Technology,
	University of Calgary, Calgary, Alberta T2N~1N4, Canada}
\date{\today}
\begin{abstract}
Constructing all extreme instances of the set of completely positive trace-preserving (CPTP) maps,
i.e., quantum channels,
is a challenging and valuable open problem in quantum information theory. 
Here we introduce a systematic approach that despite the lack of knowledge about the full parametrization of the set of CPTP maps on arbitrary Hilbert-spaced dimension, enables us to 
construct exactly those extreme channels that are covariant with respect to a finite discrete group or a compact connected Lie group. Innovative labeling of quantum channels by group representations enables us to identify the subset of group-covariant channels
whose elements are group-covariant generalized-extreme channels. 
Furthermore, we exploit essentials of group representation theory to introduce equivalence classes for the labels and also partition the set of group-covariant channels. As a result we show that it is enough to construct one representative of each partition.
We construct Kraus operators for  group-covariant generalized-extreme channels by solving systems of linear and quadratic equations for all candidates satisfying the necessary condition for being group-covariant generalized-extreme channels.
Deciding whether these constructed instances are extreme or quasi-extreme is accomplished by solving system of linear equations.
Proper labeling and partitioning the set of group-covariant channels leads to a novel systematic, algorithmic approach for constructing the entire subset of group-covariant extreme channels.
We formalize the problem of constructing and classifying group-covariant generalized extreme channels,
thereby yielding an algorithmic approach to solving,
which we express as pseudocode. 
To illustrate the application and value of our method,
we solve for explicit examples of group-covariant extreme channels. With unbounded computational resources to execute our algorithm,
our method always delivers a description of an extreme channel for any finite-dimensional Hilbert-space and furthermore guarantees a description of a group-covariant extreme channel for any dimension and for any finite-discrete or compact connected Lie group if such an extreme channel exists.
\end{abstract}
\maketitle
%\tableofcontents
\section{Introduction}
\label{sec:Intro}
%\paragraph{Motivation:}
Quantum channels, which are completely positive trace-preserving linear maps belonging to the set of linear maps on Banach spaces of operators,
represent the most general allowed form of quantum dynamics~\cite{Ludwig1968, Hellwig1969, Doplicher1971}
and they form a convex subset.
Characterizing quantum channels is important for representing general dynamics
and for modeling decoherence in quantum systems~\cite{Holevo2001}. 
Based on such characterizations, efficient simulation of quantum dynamics becomes feasible.
Another importance of characterizing the whole set of quantum channels is to describe general quantum communication channels and to analyze rates of reliable classical and quantum information that can be communicated through quantum channels~\cite{Shannon1948-1, Shannon1948-2}.

%\paragraph{Conundrum:}
As full characterization of convex sets is feasible by just knowing the extreme points of a convex set,
full channel characterization can be accomplished by determining only
the small subset comprising all extreme channels.
The conundrum is that extreme channels can be determined by knowing the convex set and vice versa. For quantum channels acting on $d$-dimensional Hilbert-space, necessary and sufficient conditions for a channel to be extreme are described in~\cite{LS93} using similar arguments to those presented in \cite{Cho75}.
For qubit channels,
full characterization of extreme points is known~\cite{FA98, FA99, RSW02, BGNPZ13}.
On the other hand,
for Hilbert-space dimension~$d>2$ (qubit channels), neither extreme channels nor full characterizations of the set of quantum channels are known

%\paragraph{Claims:}
Here we advance the understanding and characterization of extreme channels by treating channels restricted by symmetries,
which are specified by finite discrete or compact connected Lie groups.
Compactness ensures that the connected Lie group has finite-dimensional representations.
We exploit this symmetry to construct exact forms of extreme channels.
Specifically, for any~$d$,
we develop an algorithmic approach to derive exactly a subset of extreme points of the set of channels that have certain specified symmetries.

%\paragraph{Background on Mathematical Problem:}
Only special types of channels have been fully characterized to date:
a set of qubit channels including all extreme points \cite{FA98, FA99, RSW02, BGNPZ13}
and some extreme points for the set of unital channels~\cite{MW09, Haagerup2020},
which are quantum counterparts of classical bistochastic processes.
Studying unital channels is useful for investigating similarities between classical and quantum processes, 
such as establishing a quantum version of Birkhoff's Theorem~\cite{MW09}
and proving additivity or superadditivity of quantities relevant to the communication capacity of channels~\cite{PWPR2006,King2002,Fukuda2007}.
Fully characterizing quantum channels is quite challenging,
which necessitates tackling restricted cases such as unitality.

%\paragraph{Background on Quantum Simulation}
Characterizing quantum channels has an important practical application to quantum simulation~\cite{Feynman1982, Lloyd1996, BulutaNori2009}.
Quantum simulation is an important application of quantum computing and typically is studied for Hamiltonian-generated unitary evolution \cite{Lloyd1996,AharanovTa-shma, BerryAhokasSanders2007,childs2010, WiebeBerryHoyerSanders2011, Sanders2013},
but,
as general evolution is described by quantum channels,
a fully developed theory of quantum simulation could be based on simulating quantum channels.
Whereas Hamiltonian simulation exploits notions such as the Solovay-Kitaev theorem~\cite{DawsonNielsen2006, NCh10} for gate decomposition and sparseness of Hamiltonians \cite{childs2010,BerryAhokasSanders2007},
direct quantum simulation of channels is challenging. 
However some progress has been made by exploiting knowledge of extreme channels. Decomposing the single-qubit channel has been explored theoretically~\cite{WBOS13}
and experimentally~\cite{Lu2017}
and exploits properties of extreme channels.
Extreme channels are valuable as well for qudit-channel decomposition~\cite{WS15},
including for dimension-altering channels~\cite{Wang2016},
and for decomposing $m$-qubit to $n$-qubit channels~\cite{ItenColbeckChristandl2017}.
This latter result~\cite{ItenColbeckChristandl2017}
emphasizes the importance of their constructive approach to channel decomposition,
which guarantees success of the channel-decomposition procedure.
Those authors contrast their constructive approach to the qudit-channel decomposition approach~\cite{WS15},
which is provably not guaranteed to succeed based on an insufficient number of parameters. These theoretical~\cite{WBOS13, Wang2016,  ItenColbeckChristandl2017} and experimental~\cite{Lu2017} advances point to the importance of determining extreme channels to make quantum-channel simulation efficient or at least tractable.

%\paragraph{Methods:}
Solving for extreme channels is currently restricted to particular examples of unital channels,
whereas our goal is to establish a systematic, algorithmic approach to constructing extreme channels that are group-covariant~\cite{Scutaru1979}, whether unital or not.
Beginning with the name of a finite-discrete group or a compact connected Lie group and~$d$,
we exploit the possibility of being able to look up all inequivalent irreducible representations (irreps) of the group in order to be able to construct all possible inequivalent $d$-dimensional representations of the group by direct sums of inequivalent group irreps.
For any two $d$-dimensional inequivalent representations of the group and any inequivalent group irrep with dimension less than or equal to~$d$,
we solve a set of linear equations,
obtained by applying the group-covariance constraint,
to construct Kraus representations of corresponding group-covariant generalized extreme channels.
Our method exploits the full power of representation theory;
if we employed the obvious brute-force approach instead,
we would be solving a set of linear equations
for all pairs of $d$-dimensional representations
and for all $d^2$-dimensional representations of the given group,
which is an uncountably infinite number of candidates;
instead, provided the representation theory for the group is known,
our approach yields only a finite number of candidates,
making our approach feasible algorithmically.
Once we have identified these candidates,
we then test if the obtained generalized extreme channel satisfies the constraint for being extreme.
This constraint is expressed as a system of linear equations whose solution reveals whether the obtained generalized extreme channel is extreme or quasi-extreme. 

%\paragraph{Pseudocode:}

Our systematic, algorithmic approach is described by a pseudocode that we define for this purpose.
Typically, pseudocode serves as a convenient way of representing the logical flow of a program for implementation on a standard,
i.e., Turing-like,
computer,
but our pseudocode is quite different:
serving as a representation of the logical flow for our mathematical approach.
Thus, we make it clear that, formally,
our pseudocode applies to a real-number model of computing;
this model enables us to be rigorous with respect to the logic of our systematic, algorithmic approach to solving group-covariant extreme and quasi-extreme channels.

%\paragraph{Outline:}
We begin by presenting a full background to our work in~\S\ref{sec:background},
including state of the art and methods,
and then we proceed to describe our approach in~\S\ref{sec:approach}.
Our results are presented and fully explained in~\S\ref{sec:results}
followed by a discussion of these results in~\S\ref{sec:discussion}.
Finally,
we conclude in~\S\ref{sec:conclusions}
including an outlook on outstanding problems and potential future work.

%-------------------------------------------------------------------------
\section{Background}
\label{sec:background}
In this section we summarize the pertinent literature and provide basic concepts required for subsequent sections.
We begin by discussing quantum channels in~\S\ref{subsec:quantumchannels}, 
including their Kraus and Choi representations. 
\S\ref{subsec:groupcovariantchannels}
is devoted to group-covariant channels and the constraints on the Kraus operators of group-covariant quantum channels. 
%----------------------------------------------------
\subsection{Quantum channels}
\label{subsec:quantumchannels}
In this subsection first we review the definition of quantum channels. We then review Kraus representation and Choi matrix representations for quantum channels. Following that, we recall the definition of specific subsets of quantum channels, namely extreme channels, generalized-extreme channels, and quasi-extreme channels.
\subsubsection{Channel representation}
\label{subsubsec:channelrep}
For~$\mathscr{H}$ a complex finite-dimensional Hilbert space
and $\mathcal{L}(\mathscr{H})$ the space of linear operators acting on $\mathscr{H}$,
density operators~$\{\rho\}$ 
are positive trace-class operators on~$\mathscr{H}$;
i.e., they belong to the subset of $\mathcal{L}(\mathscr{H})$ denoted by
\begin{equation}
    \mathcal{T}(\mathscr{H})=\{\rho\in\mathcal{L}(\mathscr{H})| \rho\geq 0, \operatorname{tr}(\rho)=1\}.
\end{equation}
For our purposes, the trace of the density operator is unity.
A quantum channel is any completely positive trace-preserving map $\Phi:\mathcal{T}(\mathscr{H})\to\mathcal{T}(\mathscr{H})$

For~$\mathscr H$ restricted to finite dimension~$d$,
i.e.,
\begin{equation}
\label{eq:dimH}
    d:=\operatorname{dim}\mathscr{H},
\end{equation}
every quantum channel can be expressed as
\begin{equation}
\label{eq:Krausrep}
    \Phi(\bullet)
        =\sum_{k=1}^K
            A_k\bullet A_k^\dagger,\;
        \bullet\in\mathcal{T}(\mathscr{H}),\,
    K\leq d^2.
\end{equation}
This expression is subject to the trace-preserving constraint
\begin{equation}
\label{eq:TP}
    \Xi=\mathds1
\end{equation} for
\begin{equation}
\label{def:Xi}
    \Xi:=\sum_{k=1}^KA_k^\dagger A_k,
\end{equation}
which can be non-diagonal. Here the nonzero linear operators  $A_k\in\mathcal{L}(\mathscr{H})$
are called Kraus operators
with each Kraus operator expressible as a $d\times d$ complex matrix
whose entries are
\begin{equation}
\label{eq:Akij}
    (A_k)_{ij}:=\bra{e_i}A_k\ket{e_j}\in\mathbb{C},\;
    i+1,j+1\in[d]:=\{1,2,\ldots,d\},
\end{equation}
for~$\{\ket{e_i}\}$
an orthonormal basis of finite-dimensional $\mathscr{H}$, where $i$ goes from $0$ to $d-1$.
In summary, $\Phi$
can be represented by the set~$\{A_k\}_{k\in[K]}$
with each~$A_k$ comprising $d^2$ complex-valued matrix elements.
Therefore, the channel is described by up to~$Kd^2$ complex-valued parameters.
\begin{remark}
For $d=0$, the only vector in Hilbert space is zero,
which has zero norm. Hence there is no allowed state in this case.
For $d=1$, only one normalized state exists, which forms the normal basis for the Hilbert space.
Kraus operators of this channel are proportional to the projector onto the basis of the space satisfying the trace-preserving condition.
\end{remark}

The set of Kraus operators describing a map~$\Phi$ is unique up to an isometry~\cite{NCh10}.
For a given map~$\Phi$, 
the minimum number of Kraus operators,
called the Choi rank,
equals the rank of the Choi operator
\begin{equation}
\label{eq:Choiop}
    C_{\Phi}:=\frac{1}{d}
        \left(\Phi\otimes\mathds1\right)\left(\ket\Psi\bra\Psi\right),\;
    \ket{\Psi}:=\frac{1}{\sqrt{d}}\sum_{i=0}^{d-1}\ket{e_i,e_i}
\end{equation}
and
\begin{equation}
    \left(C_\Phi\right)_{m n,p q}
    =\bra{e_m,e_n}C_\Phi\ket{e_p,e_q}=\frac{1}{d}\bra{e_m}\Phi(\ket{e_n}\bra{e_q})\ket{e_p}
\end{equation}
is the Choi matrix.
Typically, the operator~(\ref{eq:Choiop})
is called the Choi matrix,
but,
due to our algorithmic approach,
we need to be extra careful in distinguishing operators from their matrix representations
and we denote
matrices of size~$m\times n$
with entries drawn from the field~$\mathbb F$
by $\mathcal{M}_{m\times n}(\mathbb{F})$.
Choi showed that a minimal set of Kraus operators
can be obtained from the eigenvectors of the Choi matrix with non-zero eigenvalus~\cite{Cho75}.
\subsubsection{Extreme channels}
\label{subsubsec:extremechannels}
The set of quantum channels~{$S_{\Phi}$}
is convex and thus has extreme points.
\begin{definition}
\label{def:extremechannels}
Extreme points of the convex set {$S_{\Phi}$} are called extreme channels.
Extreme channels are channels that cannot be written as a convex combination of any other two distinct channels in a non-trivial way.
\end{definition}
We denote the set of extreme channels by 
%$\{\Phi_\text{ext}\}\subset\{\Phi\}$. 
{$S_{\Phi_\text{ext}}\subset S_{\Phi}$}.
Despite the importance of extreme channels,
characterization of extreme channels is unknown except for the special case of $d=2$~\cite{FA98,FA99, RSW02}.
Results beyond $d=2$ are restricted to characterizing extreme points of the set of unital channels ($\Phi(\mathds1)\equiv\mathds1$), which are not necessarily extreme points of the set of all channels~\cite{MW09}.

An important theorem on extreme channels specifies necessary and sufficient conditions for a channel to be an extreme one~\cite{Cho75}:
\begin{theorem} 
\label{theorem:Choi}
A channel represented by a set of Kraus operators $\{A_k\}_{k\in[K]}$
is extreme if and only if (iff)
the set of operators
\begin{equation}
\label{eq:SAA}
    \mathcal{S}:=\{A_k^{\dag}A_l\}_{k,l\in[K]}
\end{equation}
is linearly independent~\cite{Cho75,LS93}.
\end{theorem}
\noindent
Thus, the number~$K$ of Kraus operators of an extreme channel has upper bound
\begin{equation}
\label{eq:KrausUpperBound}
    K\leq d.
\end{equation}
Therefore,
the Choi rank of an extreme channel is bounded by~$d$.

Clearly not all channels with Choi rank satisfying inequality~(\ref{eq:KrausUpperBound}) are extreme channels,
but such channels are interesting as well~\cite{R00,RSW02}.
The fact that other channels are interesting leads to defining two further important subsets of quantum channels.
One subset is known as generalized-extreme channels and the other subset is known as quasi-extreme channels,
which we now define.
\begin{definition}
Channels with Choi rank not exceeding~$d$
are called generalized-extreme channels~\cite{RSW02}
and the set of generalized-extreme channels is denoted by 
%$\{\Phi_\text{gen}\}$.
{$S_{\Phi_\text{gen}}$}
\end{definition}
\begin{definition}
 Generalized-extreme channels that are not extreme are called quasi-extreme channels~\cite{RSW02},
and the set of quasi-extreme channels is denoted by~{$S_{\Phi_\text{qe}}$}
%$\{\Phi_{\text{qe}}\}$.
\end{definition}
\begin{remark}
\label{remark:mutuallyexclusive}
Extreme and quasiextreme channels are mutually exclusive:
%$\{\Phi_\text{ext}\}\cap\{\Phi_{\text{qe}}\}=\emptyset$.
${S_{\Phi_\text{gen}}\cap S_{\Phi_{\text{qe}}}}=\emptyset$.
\end{remark}
%-----------------------------------------------------------
\subsection{Group-covariant channels}
\label{subsec:groupcovariantchannels}
In this subsection,
we elaborate on a specific class of channels,
namely group-covariant channels. First we recall the definition of equivalent channels. Based on this definition we explain that a group-covariant channel is a channel~(\ref{eq:Krausrep})
with the additional property that the channel's action is invariant under pre- and post-unitary conjugations that are described by group representations. 
Then we explain the constraints on the Kraus operators of group-covariant channels and how two 
group-covariant channels under the same group, with respect to equivalent representations of the group, are equivalent. 
Group-covariant channels have been studied in the context of channel capacity~\cite{Holevo2002, KoenigWehner2009, DattaTomamichelWilde2016,WildeTomamichelBerta2017,SiddharthaBaumlWilde2020}, extreme points of unital channels~\cite{MW09},
and channel characterization~\cite{KMM11,MozrzymasStudzinskiDatta2017,SiudzinskaChruscinski2018},
and complementarity and additivity properties of various covariant channels are discussed in~\cite{DattaFukudaHolevo2006}.
\begin{definition}
\label{def:channelequivalence}
A Channel~$\Phi$ is unitarily equivalent to channel $\Phi'$,
denoted by $\Phi\sim\Phi'$, if there exist $d$-dimensional unitary operators $U$ and $V$
 such that
 \begin{equation}
 \label{eq:equivalence}
     \Phi\sim\Phi':  \Phi=\mathcal{T}_U\circ\Phi'\circ\mathcal{T}_V
\end{equation}
where $\circ$ denotes composition of maps
and
\end{definition}
\begin{equation}
\label{eq:SQ}
    \mathcal{T}_Q:
    \mathcal{L}(\mathscr{H})\to \mathcal{L}(\mathscr{H}):
        \bullet\mapsto  Q\bullet {Q}^{\dagger},\;
    Q\in\mathcal{L}(\mathscr{H}).
\end{equation}
%Equivalency relation in Eq.~(\ref{eq:equivalence}) partitions the set of quantum channels acting on $d$-dimensional Hilbert space. 
The equivalence relation $\sim$ (\ref{eq:equivalence}) partitions the set of all quantum channels for given Hilbert-space dimension~$d$ into equivalence classes of channels. Any quantum channel is equivalent to itself for $U=V=\mathds1$. However, for some channels,
each channel is equivalent to itself even if~$U$ and~$V$ are not identity operators. These channels, with symmetric properties,
are group-covariant channels defined as below.

\begin{definition}
For a finite discrete group or a compact connected Lie group denoted~$\mathcal{G}$,
with a pair of unitary representations~\cite{Scutaru1979}
\begin{equation}
\label{eq:D1D2unitaryreps}
    D^{(1)},D^{(2)}\in\mathcal{L}(\mathscr{H}),
\end{equation}
a channel~$\Phi$ is group-covariant with respect to representations
$D^{(1)},D^{(2)}$
if 
\begin{equation}
\label{eq:GCovariant}
    \mathcal{T}_{{D}^{(2)}(g)}\circ\Phi\circ\mathcal{T}_{{D}^{(1)}(g)}=\Phi\;
    \forall g\in\mathcal{G}
\end{equation}
where $\mathcal{T}_Q$ is defined in Eq.~(\ref{eq:SQ}).
%where $\circ$ denotes composition of maps
%and
%\begin{equation}
%    \mathcal{S}_Q:
%    \mathcal{B}(\mathscr{H})\to %\mathcal{B}(\mathscr{H}):
%        \bullet\mapsto  Q\bullet %{Q}^{\dagger},\;
%    Q\in\mathcal{B}(\mathscr{H}),
%\end{equation}
\end{definition}
%We recall that as any map in $\text{GL}(\mathscr{H})$ with finite dimensional $\mathscr{H}$ is bounded, $D^{(\imath)}(g)\in\text{GL}(\mathscr{H})$, implies that $D^{(\imath)}(g)$ and hence, super operators $\mathcal{S}_{{D}^{(i)}(g)}$ are well defined. 
Suppose channel~$\Phi$,
represented by Kraus operators $\{A_k\}_{k\in[K]}$~(\ref{eq:Krausrep}), 
is group-covariant with respect to two~$d$-dimensional unitary representations of the group~$\mathcal G$,
namely,
$D^{(1)}$ and~$D^{(2)}$.
Then Eq.~(\ref{eq:GCovariant})
implies that
%\begin{align}
 %   \sum_{k}^K& \left({D^{(2)}}^{\dagger}(g)
%        A_k D^{(1)}(g)\right)\bullet  {\left({D^{(2)}}^{\dagger}(g)A_k D^{(1)}(g)%\right)}^{\dagger}
%            \nonumber\\
%    =&\sum_{k=1}^K A_k \bullet A_k^{\dagger}\;
%    \forall\bullet\in\mathcal{B}(\mathscr{H}).
%\end{align}
%Therefore, two sets of Kraus operators, $\{A_k\}_{k\in[K]}$ and~$\{{D^{(2)}}^{\dagger}(g)A_k D^{(1)}(g)\}$ describe the same channel.
%Hence, they are related by a unitary matrix $\Omega(g)\in\mathbb{C}^{K\times K}$~\cite{nielsen00}; that is
\begin{equation}
\label{eq:CovariantKraus}
    {D^{(2)}}^{\dagger}(g)A_k D^{(1)}(g)=\sum_{l=1}^K \Omega_{kl}(g)A_l,\;
    \forall k\in[K],\;
    \forall g\in\mathcal{G},
\end{equation}
subject to the trace-preserving condition~(\ref{eq:TP}),
for~$\Omega$ a unitary representation of the group~$\mathcal G$
on any~$K$-dimensional unitary space~\cite{Holevo2002}, \cite{KMM11}. The dimension of~$\Omega$ is equal to the number of Kraus operators describing the channel~\cite{KMM11}. 
\begin{remark}
\label{remark:Omega}
The group-covariant channel $\Phi$~(\ref{eq:GCovariant}) with respect to $D^{(1)}$ and $D^{(2)}$
 is not necessarily unique
 and each distinct~$\Phi$
 should be suitably labelled.
The role of $\Omega$~(\ref{eq:CovariantKraus})
is to label distinct group-covariant channels with respect to~$D^{(1)}$ and~$D^{(2)}$;
i.e.\ we can write~$\Phi_\Omega$. 

\end{remark}
\begin{remark}
\label{remark:reducibleOmega}
If~$\Omega$ is a reducible representation of the group,
then the group-covariant channel with Kraus operators satisfying Eq.~(\ref{eq:CovariantKraus}) is a convex combination of other channels that are also group-covariant with respect to~$D^{(1)}$ and~$D^{(2)}$~\cite{KMM11}. 
\end{remark}
\begin{remark}
\label{remark: UniEquivReps}
Let $\Phi$~(\ref{eq:Krausrep})
be group-covariant with respect to representations
$D^{(1)}$ and $D^{(2)}$.
Let~$D'^{(1)}$ and~$D'^{(2)}$
be two representations of the same group.
Suppose these two representations
are respectively unitarily equivalent to $D^{(1)}$ and $D^{(2)}$
so
\begin{equation}
\label{eq:equireps}
    U_iD^{(i)}(g)U_i^{\dag}=D'^{(i)}(g), \;\forall g\in\mathcal{G},
\end{equation}
for~$\{U_i\}_{i\in[2]}$ $d$-dimensional unitary operators.
Then channel $\Phi'$,
described by Kraus operators $\{A'_k=U_2A_kU_1^{\dagger}\}_{k\in[K]}$,
is group-covariant with respect to representations $D'^{(i)}$s~\cite{KMM11}
\begin{equation}
\label{eq:VPhiU}
    \Phi'=\mathcal{T}_{U_2}\circ\Phi\circ\mathcal{T}_{U_1^{\dagger}}.
\end{equation}
That is according to Definition \ref{def:channelequivalence}, these channels are unitarily equivalent. 
\end{remark}
Equality~(\ref{eq:CovariantKraus}) holds iff 
this equality holds for all generators of the group.
Consider any finite discrete group~$\mathcal{G}$ generated by 
a subset of $\mathcal{G}$, namely
\begin{equation}
\label{eq:Sgenerators}
    S=\{g_1,g_2,\ldots,g_r\},\; r:=\operatorname{rank}(\mathcal{G}).
\end{equation}
\noindent
Then Eq.~(\ref{eq:CovariantKraus}) is satisfied for all
$g\in\mathcal{G}$
if Eq.~(\ref{eq:CovariantKraus}) is satisfied for all $g_i\in S$.
Now consider any compact connected Lie group~$\mathcal{G}$
with corresponding Lie algebra
\begin{equation}
\label{eq:Liegen}
    \mathfrak{g}\in\operatorname{span}\{T_n\},\;
    n\in[\nu],\;
    \nu:=\operatorname{dim}\mathfrak{g}
\end{equation}
for~$T_n$ a generator.
Then Eq.~(\ref{eq:CovariantKraus})
can be expressed in terms of the representations of $\mathfrak{g}$~\cite{KMM11},
\begin{equation}
\label{eq:CovariantLieKraus}
    D^{(1)}(T_n)A_k-A_kD^{(2)}(T_n)=\sum_{l=1}^K {\Omega}_{kl}(T_n) A_l, \;\forall T_n,
\end{equation}
where~$D^{(1)}$ and~$D^{(2)}$
are $d$-dimensional Hermitian representations of the algebra~$\mathfrak{g}$ and~$\Omega$ is a $K$-dimensional Hermitian representation of the Lie algebra $\mathfrak{g}$.
\begin{remark}
\label{remark:commutator}
If~$D\equiv D^{(1)}=D^{(2)}$,
Eq.~(\ref{eq:CovariantLieKraus})
simplifies to the commutator 
\begin{equation}
\label{eq:CovariantLieKrauscommutator}
    \left[D(T_n),A_k\right]
    =\sum_{l=1}^K {\Omega}(T_n) A_l\;\forall T_n.
\end{equation}
\end{remark}
\begin{comment}
\begin{remark}
\label{remark: UniEquivReps}
If $\Phi$,
described by a set of Kraus operators $\{A_k\}_{k\in[K]}$~(\ref{eq:Krausrep}),
is group-covariant with respect to representations~$\{D^{(i)}\}_{i\in[2]}$,
and if $\{D'^{(i)}\}_{i\in[2]}$ 
are representations of the same group
and are unitarily equivalent to $\{D^{(i)}\}_{i\in[2]}$,
i.e.,
\begin{equation}
\label{eq:equireps}
    U_iD^{(i)}(g)U_i^{\dag}=D'^{(i)}(g), \;\forall g\in\mathcal{G},
\end{equation}
for~$\{U_i\}_{i\in[2]}$ $d$-dimensional unitary operators,
then channel $\Phi'$ described by a set of Kraus operators $\{A'_k=U_2A_kU_1^{\dagger}\}_{k\in[K]}$
is group-covariant with respect to representations $D'^{(i)}$s~\cite{KMM11}.
These two channels are related by
\begin{equation}
\label{eq:VPhiU}
    \Phi'=\mathcal{S}_{U_2}\circ\Phi\circ\mathcal{S}_{U_1^{\dagger}}.
\end{equation}
\end{remark}
\end{comment}
In this section we have reviewed the main properties of quantum channels and group-covariant channels.
Now we use these results for our aim of constructing group-covariant extreme points of the set of quantum channels~\cite{KMM11}.
%-----------------------------------------

%-----------------------------------------------------------------------------------------------
\section{Approach}
\label{sec:approach}
In this section based on the background provided in~\S\ref{sec:background} we introduce our systematic approach to construct  group-covariant generalized-extreme channels where the corresponding group has unitary representation. 
In subsection~\S\ref{subsec:formalproblems},
we describe the problem as a computational problem.
In subsection~\S\ref{subsec:mathematical}, we present a constructing subset of extreme channels that are group covariant. The next subsection~\S\ref{subsec:AlgorithmicApproach} is devoted to the algorithmic approach for solving the problem (
in~\S\ref{subsec:Pseudocode},
we explain our approach to pseudocode). 

%----------------------------------------------
\subsection{Formal problems}
\label{subsec:formalproblems}
The purpose of this subsection is to formalize problems of constructing group-covariant extreme channels
over finite-dimensional Hilbert space.
Specifically,
the construction of these channels is achieved by obtaining the exact Kraus representations for generalized-extreme group-covariant channels,
and we deal with both finite discrete and compact connected Lie groups.
Our final problem concerns deciding whether a channel 
is either extreme or not.

We formulate problems by specifying the inputs and outputs, and the problem in each case is to map the inputs to the outputs,
although the problem formulation does not strictly use this language.
Each problem is thus a task that needs to be performed to construct group-covariant extreme and quasi-extreme channels.

We now formulate and explain our three problems.
The first two problems concern construction of group-covariant generalized-extreme channels
discussed in~\S\ref{subsec:quantumchannels},
first for finite discrete groups and then for compact connected Lie groups.
For our first problem,
the input comprises the name of the finite discrete group and the dimension of the Hilbert space.
The output of the first problem is the set of $d\times d$ matrix representations of Kraus operators for group-covariant generalized extreme channels. 
Our second problem is similar to the first,
with the input comprising the name of the group,
but in this case a compact connected Lie group;
otherwise the statement is the same in that the input includes the dimension of the Hilbert space,
and the output 
is the set of $d\times d$ matrix representations of Kraus operators for group-covariant generalized extreme channels.
These first two problems are now given.
\begin{problem}
\label{prob:discrete}
Construct all exact~$d$-dimensional Kraus operators 
of all group-covariant generalized-extreme channels
for any finite discrete group.
\end{problem}
\noindent
\begin{problem}
\label{prob:Lie}
Construct all exact~$d$-dimensional Kraus representations
of all group-covariant generalized-extreme channels
for any compact connected Lie group.
\end{problem}

The third problem, which is a decision problem, 
accepts matrix representations of Kraus operators 
for a channels as input. 
The problem is to solve whether the input channel is extreme or not
and yields this answer as a single-bit output. Therefore, in our case in which the inputs are the matrix representations of Kraus operators for a generalized-extreme channel, the problem solves whether the input channel is extreme or quasi-extreme
\noindent
\begin{problem}
\label{prob:decision}
Decide whether a given set of~$d$-dimensional
Kraus operators of a quantum channel describes an extreme channel or not.
\end{problem}
\noindent
Now that we have three well-posed computational problems,
albeit permitting real and complex numbers,
we proceed to describe our approach for designing a proper procedure and its presentation as an algorithm that solves the problem. 
%-----------------------------------------------------------------------------------------------------------------
\subsection{Algebraic approach}
\label{subsec:mathematical}
Building on the formal problems posed in~\S\ref{subsec:formalproblems},
we explain our approach for constructing generalized-extreme group-covariant channels.
First {{in \S\ref{subsubsec:obtaining}}} we transform the relation between Kraus operators of a group-covariant channel to systems of linear equations. Then {{in \S\ref{subsubsec:constructinggegcchannels}}} we discuss how to construct generalized-extreme group-covariant channels. Finally, {{in \S\ref{subsubsec:equivalentextremechannels}}} we explain how defining equivalent channels helps to construct generalized-extreme group-covariant channels which only needs to be done for group-covariant channels with respect to inequivalent representations of the group. Discussions in \S\ref{subsubsec:constructinggegcchannels} and \S\ref{subsubsec:equivalentextremechannels} enables us to solve Eq.~(\ref{eq:CovariantKraus}) just for the representative of each class of group-covariant generalized extreme channels, instead of employing a brute-force approach that is constructing all group-covariant channels and then deciding which one is extreme.

\subsubsection{
Solving a system of linear equations for Kraus operators of a group-covariant channel}
\label{subsubsec:obtaining}
%\paragraph{About:}
In this subsubsection, we convert the relations for Kraus operators representing a group-covariant channel in the finite discrete case~(\ref{eq:CovariantKraus})
and in the compact connected Lie group case (\ref{eq:CovariantLieKraus})
into systems of linear equations.
The algorithm for solving this system of linear equations can then be solved algorithmically and is amenable to expressing in pseudocode.
First,
for discrete finite groups and compact connected Lie groups,
we express relations between Kraus operators as systems of linear equations.
In both cases,
we discuss instances for which 
group-covariant channels do not exist with respect to particular representations~$D^{(1)}$,~$D^{(2)}$ and $\Omega$. 

%\paragraph{Reshape Kraus operators as vectors:}
We denote the vectorized form of a Kraus operator $A_k$ by $\bm{A}_k$,
and matrix elements convert to vector elements according to $(A_k)_{ij}=(\bm{A}_k)_{id+j}$.
Concatenation of vector representations is denoted by~$\oblong$,
i.e., $\bm{A}_k\oblong\bm{A}_{k'}$
for the concatenation of vectors representing Kraus operations~$A_k$
and~$A_{k'}$, respectively.
Concatenation of a length~$K$ sequence of vectors representing Kraus operators is expressed as
\begin{equation}
\label{eq:Abigbox}
        \bm{A}:=\bigbox_{k=1}^{K}\bm{A}_k
            \in\mathcal{M}_{Kd^2\times1}(\mathbb{C}),
\end{equation}
which represents the channel as a vector comprising all of the channel's Kraus operators.

%\paragraph{Solve reshaped Kraus vectors for finite discrete group:} 
For a group-covariant channel with a finite discrete group,
a vector $\bm A$~(\ref{eq:Abigbox})
is obtained by solving
linear equations~(\ref{eq:CovariantKraus})
for each $g_i\in S$~(\ref{eq:Sgenerators}) and then imposing the trace-preserving constraint~(\ref{eq:TP}).
We re-express the linear equations~(\ref{eq:CovariantKraus})
as

\begin{equation}
\label{eq:LinearizedCovariantKrausDis}
    P(g_i)\bm{A}=\bm0\;
    \forall \,
    i\in[r],
\end{equation}

for
\begin{equation}
\label{eq:P(g)}
    P(g_i)=\bigoplus_{k=1}^{K} \left({D^{(2)}}^{\dag}(g_i)\otimes {D^{(1)}}^\text{T}(g_i)\right)-\Omega(g_i)\otimes \mathds1_{d^2}
    \in\mathcal{M}_{Kd^2\times Kd^2}(\mathbb{C})
\end{equation}
with~$\mathds1_{d^2}$ the $d^2\times d^2$ identity matrix and $\text{T}$ denoting matrix transposition: $(\bullet)_{i,j}^\text{T}=(\bullet)_{j,i}$.
Equation~(\ref{eq:LinearizedCovariantKrausDis}) is a set of systems of homogeneous linear equations.
Each of these systems of homogeneous linear equations is labelled in terms of the same three representations of the group,
namely $D^{(1)}$, $D^{(2)}$ and $\Omega$,
as clearly seen in Eq.~(\ref{eq:P(g)}).
\begin{remark}
\label{remark:trivialcase}
In the trivial case that
${D^{(1)}}=\mathds1_d={D^{(2)}}$
and~$\Omega=1$,
then $P(g_i)\equiv0$ for all~$i$,
which implies that~$\bm A$ in Eq.~(\ref{eq:P(g)}) is unconstrained and is a vectorized version of just one single Kraus operator. 
\end{remark}

The solution $\bm{A}$~(\ref{eq:LinearizedCovariantKrausDis}) belongs to ker$(P(g_i))$ for all $i\in[r]$.
If $\bigcap_{i\in[r]}\operatorname{ker}\left(P(g_i)\right)\equiv{\{}\bm{0}{\}}$,
then the only solution to Eq.~(\ref{eq:LinearizedCovariantKrausDis}) is the trivial solution $\bm{A}\equiv0$.
Such a trivial case arises 
either if $g_i\in S$ exists such that $\det\left(P(g_i)\right)\neq 0$, which means $\text{ker}(P(g_i))={\{}\bm{0}{\}}$ for some~$i$,
or else $\text{ker}(P(g_i))\neq{\{}\bm{0}{\}}$ for all $i\in[r]$ but their intersection is zero.

This trivial solution $\bm{A}\equiv0$ 
implies non-existence of a group-covariant channel with respect to group 
representations~$D^{(1)}$,~$D^{(2)}$ and~$\Omega$. 
On the other hand,
if $\det\left(P(g_i)\right)= 0$ for all $g_i\in S$ and $\bigcap_{i\in[r]}\operatorname{ker}\left(P(g_i)\right)\neq{\{}\bm{0}{\}}$, then Eq.~(\ref{eq:LinearizedCovariantKrausDis}) yields a nontrivial solution $\bm{A}\in\bigcap_{i\in[r]}\operatorname{ker}\left(P(g_i)\right)$
with its number of parameters being less than or equal to~$\min_{i\in[r]}\{\operatorname{nullity}(P(g_i))\}$.
This non-trivial solution~$\bm{A}$
could represent valid Kraus operators of a group-covariant completely positive (CP) map with respect to given~$D^{(1)}$, ~$D^{(2)}$ and~$\Omega$
and these nontrivial solutions are candidates for solutions for a channel if the trace-preserving condition,
discussed below,
can be imposed successfully.

These algebraic equations and arguments can be understood geometrically as well,
which provides an alternative,
and valuable,
insight.
From Eq.~(\ref{eq:LinearizedCovariantKrausDis})
we know that~$\bm A$,
for each~$i$,
is in the kernel of a matrix with~$Kd^2$ rows of~$Kd^2$ elements per row.
Each nonzero row of $P(g_i)\bm{A}=\bm 0$~(\ref{eq:LinearizedCovariantKrausDis})
defines a hyperplane in $\mathbb{C}^{Kd^2}$.
Thus, the solution to Eq.~(\ref{eq:LinearizedCovariantKrausDis}) for each~$i$ is a flat, that is, the intersection of~$Kd^2$ hyperplanes obtained from rows of Eq.~(\ref{eq:LinearizedCovariantKrausDis})
with the dimension of this intersection being  nullity$(P(g_i))$.
According to  Eq.~(\ref{eq:LinearizedCovariantKrausDis}),
$\bm{A}$ belongs to the intersection
of all~$r$ flats.
If the intersection of these flats is empty, then a group-covariant channel labelled by~$D^{(1)}$, $D^{{(2)}}$ and~$\Omega$ does not exist.
If the intersection is another flat, its dimension is less than or equal to~$\min_{i\in[r]}\{\operatorname{nullity}(P(g_i))\}$,
which equals the number of parameter in the solution~$\bm{A}$. These solutions represent valid Kraus operators of a group-covariant CP map with respect to given~$D^{(1)}$, $D^{(2)}$ and~$\Omega$,
and next we constrain these solutions by the trace-preserving condition to obtain,
if possible,
solutions~$\bm A$ for a channel.

%\paragraph{Constrain via trace-preserving condition:}
For a channel,
these Kraus operator solutions for CP maps must further satisfy the trace-preserving condition~(\ref{eq:TP}),
which involves additional equations
that restrict the parameter-space domain for Kraus operators.
To impose this trace-preserving constraint,
we reshape Kraus vectors~$\{\bm{A}_k\}$
back to Kraus matrices~$\{A_k\}$. Then we verify if~$\{A_k\}$, describe a trace preserving map for a range of complex parameters in~$\{A_k\}$.
Each diagonal element of~$\Xi$~(\ref{eq:TP})
is the sum of modulus square of complex parameters of Kraus operators
\begin{equation}
\label{eq:summodsquareKrausops}
    \bra{e_i} \Xi\ket{e_i}
        =\Xi_{ii}
        =\sum_{j=1}^d\sum_{k=1}^K
        \left|(A_k)_{ij}\right|^2,
\end{equation}
hence a non-negative real number,
for~$\{\ket{e_i}\}$ the orthonormal basis of~$\mathscr{H}$. 

%\paragraph{Diagonal-matrix case:}
To constrain solutions~$\bm A$ that yield CP maps,
by applying the trace-preserving condition,
we test if~$\Xi$ is diagonal
by checking that all off-diagonal terms for this matrix are zero. 
If the test shows that~$\Xi$ is a $d$-dimensional diagonal matrix,
then we have~$d$ equations 
that depend only on modulus square complex parameters of Kraus operators.
The number of parameters
varies depending on the particular case.
Equations that are linear in modulus square of parameters of Kraus operators
can be solved by a linear equation solver for modulus square complex parameters.
The solution describes a family of group-covariant channels.

The diagonal-$\Xi$ matrix case for imposing the trace-preserving condition admits a beautiful geometric analogy.
Recalling the geometric perspective,
Kraus operators of a group-covariant CP map belongs to a flat in parameter space.
If~$\Xi$ is diagonal,
the trace-preserving constraint~(\ref{eq:TP})
is represented geometrically by at most~$d$ hyperspheres embedded in that flat
such as a circle embedded in a plane.
Thus, the intersection of these hyperspheres
determines the parameter domain for which the CP map is trace-preserving, hence a channel.

%\paragraph{Nondiagonal-matrix case:}
If~$\Xi$ is not diagonal, then, in general,
we have $d(d+1)/2$ expressions quadratic in complex parameters.
Not all these $d(d+1)/2$ are necessarily mutually independent.
Furthermore,
the number of parameters can be fewer than the maximum of~$Kd^2$.
The off-diagonal case is harder than the diagonal case with respect to imposing the trace-preserving condition,
and the geometric perspective is not as helpful so we discuss the non-diagonal case  only from an algebraic perspective.
These equations can be solved algorithmically,
for example using SimPy discussed in~\S\ref{subsubsec:blibrary},
but such solvers are not guaranteed to solve nor is a solution known to exist in general.

%\paragraph{Linearized for Lie group:}
For compact connected Lie groups, the relation between Kraus operators
of the channel~(\ref{eq:CovariantLieKraus}) 
is transformed to the set of linear equations
\begin{equation}
\label{eq:LinearizedCovariantKrausLie}
    Q(T_n)\bm{A}=\bm0\; \forall n\in[\nu],
\end{equation}
analogous to Eq.~(\ref{eq:LinearizedCovariantKrausDis}) for discrete groups,
with~$\bm{A}$~(\ref{eq:Abigbox})
and~$\nu$~(\ref{eq:Liegen}).
In Eq.~(\ref{eq:LinearizedCovariantKrausLie})
\begin{equation}
\label{eq:Q(g)}
    Q(T_n)=\bigoplus_{k=1}^{K} \left({D^{(1)}}(T_n)\otimes \mathds1_d-\mathds1_d\otimes {{D^{(2)}(T_n)}}^{\text{T}}\right)-\Omega(T_n)\otimes \mathds1_{d^2}
    \in\mathcal{M}_{Kd^2\times Kd^2}(\mathbb{C}).
\end{equation}
Equation~(\ref{eq:LinearizedCovariantKrausLie}) is a set of systems of homogeneous linear equations.
Each of these systems of homogeneous linear equations is labelled in terms of the same three representations of the group,
namely $D^{(1)}$, $D^{(2)}$ and~$\Omega$,
as clearly seen in Eq.~(\ref{eq:Q(g)}).
Algebraic and geometrical descriptions of the solution to Eq.~(\ref{eq:LinearizedCovariantKrausLie})
are similar to those for the solution of Eq.~(\ref{eq:LinearizedCovariantKrausDis}). 
As for the finite discrete-group case above,
the solution~$\bm{A}$ of Eq.~(\ref{eq:LinearizedCovariantKrausLie})
belongs to $\bigcap_{n\in[\nu]} \operatorname{ker}(Q(T_n))$.
If $\bigcap_{n\in[\nu]} \operatorname{ker}(Q(T_n))={\{}\bm{0}{\}}$,
then the system of linear equations~(\ref{eq:LinearizedCovariantKrausLie}) has only a trivial solution,
which means that the group-covariant CP map with respect to~$D^{(1)}$, $D^{(2)}$ and $\Omega$ does not exist. If, on the other hand,
$\operatorname{ker}(Q(T_n))\neq{\{}\bm{0}{\}}$,
then the solution of Eqs.~(\ref{eq:LinearizedCovariantKrausLie}) yields all Kraus operators of the group-covariant CP map.
If the trace-preserving condition can be imposed successfully,
then the group-covariant channel with respect to~$D^{(1)}$, $D^{(2)}$ and $\Omega$ is obtained following the approach discussed above for discrete groups.
\subsubsection{Constructing generalized-extreme group-covariant channels}
\label{subsubsec:constructinggegcchannels}
In this subsubsection, we explain our approach for constructing group-covariant generalized-extreme channels. 
First we establish our notation and explain how we label different group-covariant channels. Then we discuss the labels of group-covariant generalized-extreme channels. 

To construct a set of group-covariant channels 
given group~$\mathcal{G}$ and Hilbert-space dimension~$d$,
one chooses two $d$-dimensional representations~$D^{(1)}$ and~$D^{(2)}$ for~$\mathcal G$.
The set of group-covariant channels with respect to~$D^{(1)}$ and~$D^{(2)}$ is denoted by $\mathcal{W}_{\mathcal{G}, D^{(1)}, D^{(2)}}$. 
To construct~$\mathcal{W}_{\mathcal{G}, D^{(1)},D^{(2)}}$,
one selects each~$\Omega$
from the set of unitary representations for~$\mathcal{G}$
and solves the linear equations in Eq.~(\ref{eq:CovariantKraus})
to obtain matrix descriptions of Kraus operators. As explained in~\S\ref{subsubsec:obtaining},
solving Eq.~(\ref{eq:CovariantKraus}) is equivalent to solving a set of systems of homogeneous linear equations~(\ref{eq:LinearizedCovariantKrausDis})
for the case of finite discrete groups,
 which is a set of systems of homogeneous linear equations labelled by three representation of the group,
namely, $D^{(1)},D^{(2)}$ and $\Omega$.
Hence, we label the solution to this set of systems of homogeneous linear equations by the same labels and denote this solution by $\Phi_{D^{(1)},D^{(2)},\Omega}$.
This approach,
involving equivalence between Eqs.~(\ref{eq:CovariantKraus})
and~(\ref{eq:LinearizedCovariantKrausLie}) and labelling by $D^{(1)},D^{(2)}$ and $\Omega$,
pertains as well to the compact connected Lie group case.

Trivial solutions ($\Phi\equiv0$) are excluded and each non-trivial solution is denoted by $\Phi_{D^{(1)},D^{(2)},\Omega}\in\mathcal{W}_{\mathcal{G}, D^{(1)},D^{(2)}}$
the index $\Omega$ labels distinct group-covariant channels with respect to $D^{(1)}$ and $D^{(2)}$, which are general $d$-dimensional representations of $\mathcal{G}$ including the reducible case.
The set of all group-covariant channels is denoted by
\begin{equation}
\label{eq:W_G}
\mathcal{W}_{\mathcal{G}}:=\bigcup_{D^{(1)},D^{(2)}}\mathcal{W}_{\mathcal{G}, D^{(1)}, D^{(2)}},
\end{equation}
where the union is over all $d$-dimensional representations of $\mathcal{G}$. 

The set of indices labeling distinct group-covariant channels with respect to $D^{(1)}$ and $D^{(2)}$
is denoted by
\begin{equation}
    \mathcal{F}_{\mathcal{G},D^{(1)},D^{(2)}}:=\{\Omega:\Phi_{D^{(1)}, D^{(2)}, \Omega}\in\mathcal{W}_{\mathcal{G}, D^{(1)},D^{(2)}}\}.  
\end{equation}
Then $\mathcal{F}_{\mathcal{G}}:=\bigcup_{D^{(1)},D^{(2)}}\mathcal{F}_{\mathcal{G}, D^{(1)}, D^{(2)}}$ is the set of all labels of group-covariant channels. 

Unitary conjugation is an equivalence relation~$\approx$
that partitions $\mathcal{F}_{\mathcal{G}}$ to equivalence classes and each class is denoted by $[\Omega]$.
Indeed labels of classes are inequivalent representation of the group.
We show that all elements of each class~$[\Omega]$
label the same group-covariant channel;
hence,
instead of labeling distinct group representations by a unitary representation of the group, we use a more appropriate notation;
i.e., we label each group-covariant channel by an equivalence class  $\Phi_{D^{(1)},D^{(2)},[\Omega]}$.

We denote the set of group-covariant generalized-extreme channels by $\mathcal{W}_{\mathcal{G},\text{gen}}\subset\mathcal{W}_{\mathcal{G}}$
which is a subset of all generalized-extreme channels: $\mathcal{W}_{\mathcal{G},\text{gen}}\subset {S_{\Phi_{\text{gen}}}}$
and has two exclusive subsets, namely quasi-extreme group-covariant channels and extreme group-covariant channels.
This set of channels are respectively denoted by
$\mathcal{W}_{\mathcal{G},\text{qe}}\subset\mathcal{W}_{\mathcal{G},\text{gen}}$ and $\mathcal{W}_{\mathcal{G},\text{ext}}\subset\mathcal{W}_{\mathcal{G},\text{gen}}$
with $\mathcal{W}_{\mathcal{G},\text{qe}}\bigcap\mathcal{W}_{\mathcal{G},\text{ext}}=\emptyset$.
We show that the subset of  $\mathcal{F}_{\mathcal{G}}/\!\approx$,
denoted by $\mathcal{F}_{\mathcal{G},\text{gen}}$,
which includes all irreducible $[\Omega]$s with dimension not exceeding $d$, labels all channels in $\mathcal{W}_{\mathcal{G},\text{gen}}$.

Thus, to construct group-covariant generalized-extreme channels, instead of solving Eq.~(\ref{eq:CovariantKraus}) for all $\Omega$ which is a brute-force approach, we solve Eq.~(\ref{eq:CovariantKraus}) for all $[\Omega]\in\mathcal{F}_{\mathcal{G},\text{gen}}$ to obtain Kraus operators of all group-covariant generalized-extreme channels. Then by testing whether
$\mathcal S$~(\ref{eq:SAA}) is a set of linearly independent operators, we classify,
according to Theorem~\ref{theorem:Choi},
resultant non-trivial channels into extreme and quasi-extreme classes. 
\subsubsection{Equivalent extreme channels}
\label{subsubsec:equivalentextremechannels}
In this subsection we partition the set of group-covariant channels and explain that it is enough to construct one representative of each class, 
which is more efficient than the direct approach of solving every extreme channel.
We discuss that,
if an element of a class of equivalent channels as described in Definition \ref{def:channelequivalence} is an extreme group-covariant channel,
then all elements of that class have that property. 

First we show that,
according to the equivalence relation in~(\ref{eq:equivalence}), all $\Phi_{D^{(1)},D^{(2)},[\Omega]}\in\mathcal{W}_{\mathcal{G},D^{(1)},D^{(2)}}$ and all $\Phi_{D'^{(1)},D'^{(2)},[\Omega]}\in\mathcal{W}_{\mathcal{G},D'^{(1)},D'^{(2)}}$ with $D^{(i)}$ and $D'^{(i)}$
being two unitarily equivalent representation of $\mathcal{G}$~(\ref{eq:equireps}), are equivalent. We denote the set of group-covariant channels with respect to all representations of $\mathcal{G}$ that are equivalent to~$D^{(1)}$ and~$D^{(2)}$ by $\mathcal{W}_{\mathcal{G},[D^{(1)}],[D^{(2)}]}$, 
which is an equivalence class.
Then we show that if  $\Phi_{D'^{(1)},D'^{(2)},[\Omega]}\in\mathcal{W}_{\mathcal{G},[D^{(1)}],[D^{(2)}]}$ is extreme, all channels in $\mathcal{W}_{\mathcal{G},[D^{(1)}],[D^{(2)}]}$ are extreme. 
Hence, for $[\Omega]\in\mathcal{F}_{\mathcal{G},\text{gen}}$ (defined in~\S\ref{subsubsec:constructinggegcchannels})
if
\begin{equation}
    \Phi_{D'^{(1)},D'^{(2)},[\Omega]}\in\mathcal{W}_{\mathcal{G},[D^{(1)}],[D^{(2)}]}
\end{equation}
is extreme/quasi-extreme,
then all channels in $\mathcal{W}_{\mathcal{G},[D^{(1)}],[D^{(2)}]}$ 
are extreme/quasi-extreme.
Hence, to construct all generalized extreme group-covariant channels, for each $[\Omega]\in\mathcal{F}_{\mathcal{G},\text{gen}}$,
we construct~$\Phi_{D^{(1)}, D^{(2)},[\Omega]}$.
Other elements of the classes are derived according to equivalence relation among quantum channels which is any arbitrary pre-post unitary conjugation. This step is repeated for all inequivalent representations $D^{(1)}$ and $D^{(2)}$. 

%\subsubsection{Constructing Kraus representations of all group-covariant extreme channels
%and some quasi-extreme channels}
%------------------------------------------------------------------------------------------------------

\subsection{Algorithmic approach}
\label{subsec:AlgorithmicApproach}
In this subsection we explain the algorithmic approach for constructing group-covariant extreme channels for finite discrete groups and compact connected Lie groups. 
%We commence by discussing the computer library of our new functions in~\S\ref{subsubsec:alibrary},
%which augments the set of known functions in public libraries discussed in~\S\ref{subsubsec:blibrary}.
{ We commence by constructing} the algorithm for solving Problems~\ref{prob:discrete}
and~\ref{prob:Lie} in~\S\ref{subsubsec:algdiscrete}.
%Finally, 
We discuss constructing an Algorithm for Problem~\ref{prob:decision} in~\S\ref{subsubsec:algdecision}.
We express the algorithm as input and output,
including the type of each input and output
(such as symbol, integer, real number, bit, character and so on).
In our algorithmic approach,
we do not restrict ourselves to discrete mathematics;
rather we permit symbols and real- and complex-number entries in the register as we are focused on an algorithmic approach to the problem but not issues of computability or complexity associated with various computational models~\cite{AB09}.
In each procedure,
we employ required functions from our libraries
discussed in~\S\ref{subsubsec:library}.

\subsubsection{Algorithm for solving Problems~\ref{prob:discrete} and \ref{prob:Lie}}
\label{subsubsec:algdiscrete}
In this subsubsection,
we present our approach to developing the algorithm for solving Problem~\ref{prob:discrete} and \ref{prob:Lie}.
Specifically, we state the input, output and brief description of the procedure,
with the full explanation of the algorithm in the first algorithm of~\S\ref{subsec:pseudocoding}.
\paragraph{Input:}
A binary number that flags the type of the group, which is either a finite discrete group or a compact connected Lie group.
The name of any finite discrete group or a compact connected Lie group
expressed as a character string.
Examples of finite discrete group names include~$\mathbb{Z}_2$,
which is the cyclic group of order~2,
or~$S_3$,
which is the symmetric group of degree~3.
Hilbert-space dimension $d\in\mathbb{Z}^+$
is the other input.
\paragraph{Output:}
A finite number $C\in\mathbb{Z}^+$ is the first output representing the total number of generalized extreme group-covariant channels on~$d$-dimensional Hilbert space and the second output is the set of Kraus operators for all $C$ channels.
\paragraph{Procedure:}
We import required functions from the libraries discussed in~\S\ref{subsubsec:blibrary} and~\S\ref{subsubsec:rlibrary}
and declare necessary variables for the algorithm.
Then we solve the system of linear equations~(\ref{eq:LinearizedCovariantKrausDis}) and  (\ref{eq:LinearizedCovariantKrausLie}), respectively, for finite discrete group (Problem~\ref{prob:discrete}) and compact connected Lie group (Problem~\ref{prob:Lie}). In either case we solve the corresponding set of linear equations for each instance~$D^{(1)}$,~$D^{(2)}$ (inequivalent $d$-dimensional representations of the group) and~$\Omega$ (inequivalent irreps of the group with dimension less than or equal to $d$). 
\begin{itemize}
    \item 

If a non-trivial symbolic solution exists for this instance,
then this solution is unique,
and the algorithm constructs a set of Kraus operators for this solution
and imposes the trace-preserving condition on it,
which as explained in~\S\ref{subsubsec:obtaining} is done by first constructing $\Xi$~(\ref{def:Xi}) and second solving $\Xi=\mathds 1$.
If~$\Xi$ is diagonal, then a linear-solver algorithm solves a system of linear equations for modulus squares of symbols for Kraus matrices.
If the solution exists, it is stored and the counter $C$ for successful case increments.
Then the algorithm proceeds to the next instance. 
If $\Xi$ is not diagonal, as discussed in~\S\ref{subsubsec:obtaining},
a solver such as Python's SimPy aims to solve sets of quadratic equations.
If the solution is found within the allowed time,
the solution is written to the register,
the counter~$C$ increments,
and the algorithm proceeds to the next instance.
\item If the system of linear equation does not have a non-trivial solution, the algorithm goes to the next incident without incrementing the counter~$C$.
\end{itemize}
After testing all instances,
and generating a parametric solution to a system of linear equations for all instances where this solution exists,
an ordered list (with ordering determined by \textsc{gProps} defined in~\S\ref{subsubsec:rlibrary}) of all sets of Kraus operators is returned.
The algorithm returns finite number $C$ in the output.
If $C=0$,
then acceptable channels have not been found. Fig.~(\ref{fig:schem}) shows a schematic representation of algorithm for solving Problem \ref{prob:discrete} which is useful for going through the details of pseudocode in~\S\ref{subsec:pseudocoding}.
\begin{figure}
    \centering
    \includegraphics[width=\textwidth]{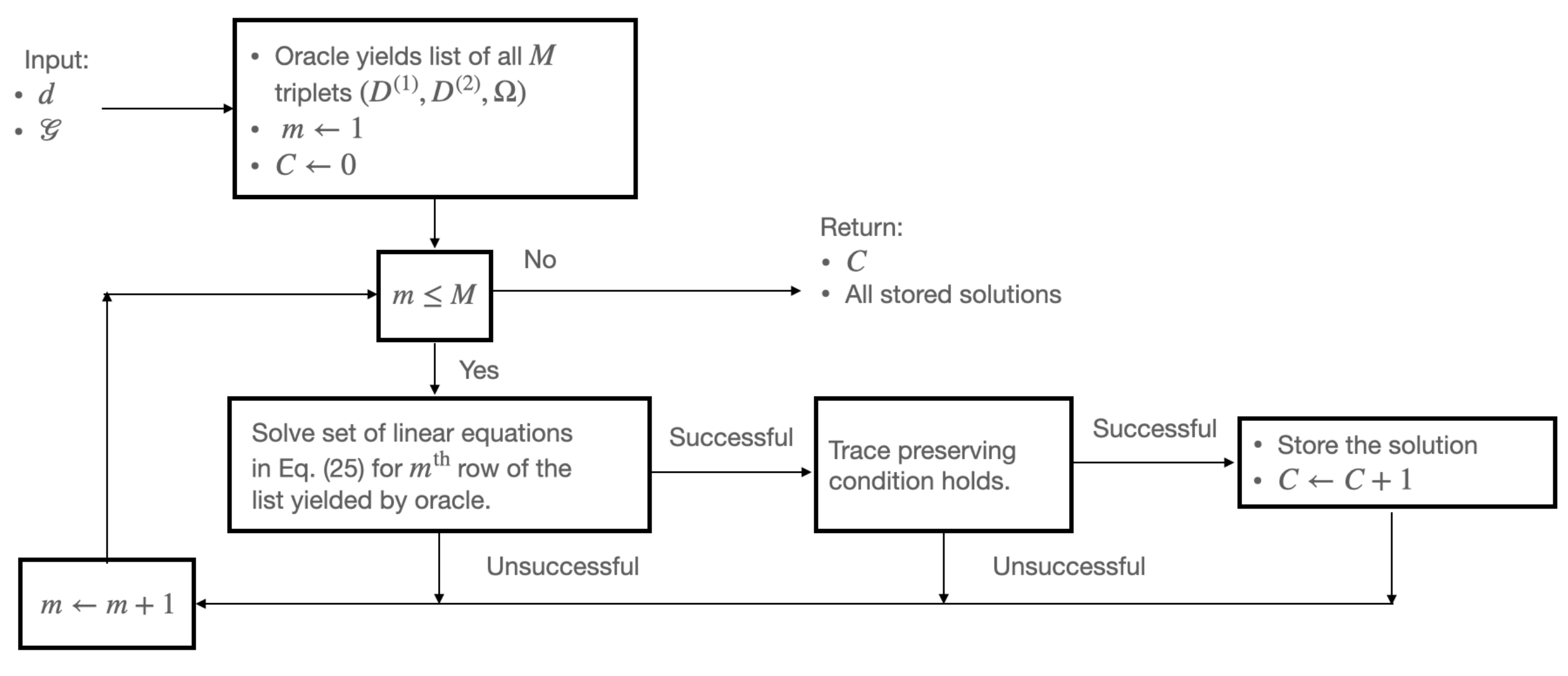}
    \caption{Schematic representation of the algorithmic approach for solving Problem~\ref{prob:discrete}}
    \label{fig:schem}
\end{figure}
%----------------------------------------------------------------
%%%%%%%%%%%%
\subsubsection{Algorithm for solving Problem~\ref{prob:decision}}
\label{subsubsec:algdecision}
In this subsubsection we describe inputs and outputs for the algorithm to solve Problem~\ref{prob:decision}.
When we present our procedure to determine whether a given channel is extreme or not. 
\paragraph{Input:}
A list of $K$ $d$-dimensional Kraus operators $\left\{A_{k\in[K]}\right\}$ with symbolic elements.
\paragraph{Output:}
The single-bit output is $\top$, denoting `true',
if the given Kraus operators correspond to an extreme channel;
otherwise the output is $\bot$,
denoting `false',
which means that the Kraus operators do not describe an extreme channel.
\paragraph{Procedure:}
We import required functions from the libraries~\S\ref{subsubsec:blibrary} and~\S\ref{subsubsec:rlibrary}
and declare necessary variables for the algorithm.
Then, from~$K$~$d$-dimensional symbolic matrices in the input $A_n$, $n\in[K]$, we construct a larger set of $K^2$~$d$-dimensional symbolic matrices by computing $B_{nn'}:=A_n^{\dagger} A_{n'}$ for all $n,n'\in[K]$. In the next step, the algorithm decides whether or not matrices $B_{nn'}$ are linearly independent.
Like typical algorithms for deciding linear independence,
our algorithm solves a system of linear homogeneous equations.
 If the only solution is trivial,
then the output is $\top$,
which means that the input channel is an extreme channel; otherwise the output is $\bot$ which means that the input channel is not extreme.
%------------------------------------------------------------
%------------------------------------------
\section{Results}
\label{sec:results}
In this section we present our new results.
First we begin in~\S\ref{subsec:constructinggceca} with establishing our method for constructing group-covariant generalized-extreme channels  for finite discrete and compact connected Lie groups and for given finite Hilbert-space dimension.
Then, in~\S\ref{subsec:pseudocoding},
we present our pseudocode for constructing generalized extreme group-covariant channels.
Finally,
in~\S\ref{subsec:explicitexamples},
we solve and elaborate on four finite discrete-group and two compact connected Lie group examples of group-covariant generalized-extreme channels
as interesting in their own right but also excellent illustrations of the generality of our method.
\subsection{Algebraic construction of group-covariant extreme channels}
\label{subsec:constructinggceca}
In this subsection we establish a procedure for constructing group-covariant generalized-extreme channels for given finite discrete and compact connected Lie groups and also given a finite Hilbert-space dimension.
We use the equivalence relation in Definition~\ref{def:channelequivalence} to show that either all elements of an equivalence class are extreme or else none are extreme.
This property of classes simplifies the procedure for constructing all generalized extreme group-covariant channels.

\subsubsection{Labeling group-covariant channels}
\label{subsubsec:rGenExtGCovariantChannels}
In this subsubsection,
we explain how to construct
the set of group-covariant generalized-extreme channel labels~$\mathcal{F}_{\mathcal{G},\text{gen}}$,
which is explained in~\S\ref{subsubsec:constructinggegcchannels}.
An element~$[\Omega]\in\mathcal{F}_{\mathcal{G},\text{gen}}$, where $\mathcal{F}_{\mathcal{G},\text{gen}}$ is defined in~\S\ref{subsubsec:constructinggegcchannels} and discussed more in~\S\ref{subsubsec:equivalentextremechannels}, labels group-covariant generalized-extreme channels which are elements of the set $\mathcal{W}_{\mathcal{G},\text{gen}}$, defined in~\S\ref{subsubsec:constructinggegcchannels},
with~$\mathcal G$ the name of the group.
We propose two lemmas 
that are needed to construct group-covariant generalized-extreme channels.

The following lemma shows that (matrix) labels~$\Omega$ and~$\Omega'$,
which are unitarily equivalent,
label the same group-covariant channel;
i.e.,
\begin{equation}
\label{eq:PhiOmegaOmega'}
    \Phi_{D^{(1)},D^{(2)},\Omega}
        =\Phi_{D^{(1)},D^{(2)},\Omega'}.
\end{equation}
\begin{lemma}
\label{lemma:equivalentOmega}
All unitarily equivalent representations of $\mathcal{G}$
label the same group-covariant channel.
\end{lemma}
\begin{proof}
We being by recognizing that Eq.~(\ref{eq:CovariantKraus}) holds for
Kraus operators~$\{A_k\}_{k\in[K]}$
if 
\begin{equation}
\label{eq:OmegaOmega'}
    \Phi_{D^{(1)},D^{(2)},\Omega}\in\mathcal{W}_{\mathcal{G}, D^{(1)}, D^{(2)}}
\end{equation}
for each~$\Phi_{D^{(1)},D^{(2)},\Omega}$
described by this set~$\{A_k\}_{k\in[K]}$~(\ref{eq:Krausrep}).
If $\Omega'$ is a representation of $\mathcal{G}$ that is unitarily equivalent to~$\Omega$ through unitary conjugation by $U$,
i.e.,
\begin{equation}
\Omega'(g)=U\Omega(g) U^{\dagger},\; \forall g\in\mathcal{G}, 
\end{equation}
then,
for
\begin{equation}
\label{eq:CA}
    A'_\ell:=\sum_{k=1}^K U_{\ell k}A_k\;
    \forall\ell\in[K]
\end{equation}
we obtain
\begin{equation}
\label{eq:ChannelC}
    {D^{(2)}}^{\dagger}(g)A'_\ell D^{(1)}(g)
    =\sum_{\ell'=1}^K \Omega'_{\ell\ell'}(g)A'_{\ell'},\; \forall g\in\mathcal{G}.
\end{equation}
%By replacing $\Omega(g)$ with $\tilde{\Omega}(g)$ in Eq.~(\ref{eq:CovariantKraus}), it is easy to see that:
%\begin{equation}\label{eq:CovaTil}
% {D^{(2)}}^{\dagger}(g)A_k D^{(1)}(g)=\sum_{l=1}^K %U^{\dagger}_{k,m}\tilde{\Omega}_{m,n}(g)C_n.
%\end{equation} 
%By multiplying both sides of Eq.~(\ref{eq:CovaTil}) with $U_{j,k}$ and summing over~$K$ it follows that
Equation~(\ref{eq:ChannelC}) implies that channel $\Phi_{D^{(1)},D^{(2)},\Omega'}$,
described by Kraus operators $\{A'_\ell\}_{\ell\in[K]}$, belongs to $\mathcal{W}_{\mathcal{G},D^{(1)},D^{(2)}}$. 
Therefore, from Eq.~(\ref{eq:CA}) we verify Eq.~(\ref{eq:PhiOmegaOmega'}).
\end{proof}
Equation~(\ref{eq:OmegaOmega'}) shows that the triplets $(D^{(1)},D^{(2)},\Omega)$ and $(D^{(1)},D^{(2)},\Omega')$ label the same group-covariant channel if $\Omega$ and $\Omega'$ are unitarily equivalent.
Unitary conjugation is an equivalence relation~$\!$ 
that partitions the set of labels of group-covariant channels $\mathcal{F}_{\mathcal{G}}$ (defined in \S\ref{subsubsec:constructinggegcchannels}) into equivalence classes denoted by $[\Omega]$. Distinct classes are represented by inequivalent representations of $\mathcal{G}$.
According to Lemma~\ref{lemma:equivalentOmega} all elements of each class~$[\Omega]$ label the same group-covariant channel. Hence, we label group-covariant channels by equivalence classes: $\Phi_{D^{(1)},D^{(2)},[\Omega]}$. 
As a consequence of Lemma~\ref{lemma:equivalentOmega}, for constructing group-covariant channels, instead of solving Eq.~(\ref{eq:CovariantKraus}) for all $\Omega$s,
just solving Eq.~(\ref{eq:CovariantKraus}) for inequivalent $\Omega$s,
which represent distinct classes $[\Omega]$s,
suffices.

The second lemma shows a necessary condition for $[\Omega]$ to label an extreme group-covariant channel.
\begin{lemma}
\label{lemma:NecCondOmega}
If
\begin{equation}
\label{eq:PhiClassOmega}
    \Phi_{D^{(1)},D^{(2)},[\Omega]}\in\mathcal{W}_{\mathcal{G},\text{ext}}
\end{equation}
then $[\Omega]$
is any class of irreps of $\mathcal{G}$ with dimension not exceeding~$d$. 
\end{lemma}
\begin{proof}
If Eq.~({\ref{eq:PhiClassOmega}}) holds, then 
according to Remark~\ref{remark:reducibleOmega},
$[\Omega]$ is a class of irreducible representation of $\mathcal{G}$ and,
due to Inequality~(\ref{eq:KrausUpperBound}),
the number of Kraus operators of $\Phi_{D^{(1)},D^{(2)},[\Omega]}$ does not exceed~$d$.
On the other hand, according to Eq.~(\ref{eq:CovariantKraus}), 
the number~$K$ of  Kraus operators of $\Phi_{D^{(1)},D^{(2)},[\Omega]}$
equals the dimension of the representation of the group denoted by~$\Omega$.
Therefore, for all  $\Phi_{D^{(1)},D^{(2)},[\Omega]}$, $[\Omega]$s are classes of irreducible representations of the group with dimension less than or equal to~$d$. 
\end{proof} 

Lemma~\ref{lemma:NecCondOmega} yields necessary but not sufficient conditions for $\Phi_{D^{(1)},D^{(2)},[\Omega]}$ to be extreme. 
If,
for $\Phi_{D^{(1)},D^{(2)},[\Omega]}\in\mathcal{W}_{\mathcal{G},\text{ext}}$, 
$\Omega$ is an irreducible representation of the group with dimension not exceeding $d$, 
and recalling that the number of Kraus operators of $\Phi_{D^{(1)},D^{(2)},[\Omega]}$ is equal to the dimension of~$\Omega$, 
one concludes that the number of Kraus operators of  $\Phi_{D^{(1)},D^{(2)},[\Omega]}$ does not exceed~$d$. 
Therefore, the Choi rank of $\Phi_{D^{(1)},D^{(2)},[\Omega]}$ is less than or equal to~$d$.
Hence, $\Phi_{D^{(1)},D^{(2)},[\Omega]}$ is a generalized-extreme channel: $\Phi_{D^{(1)},D^{(2)},[\Omega]}\in\mathcal{W}_{\mathcal{G},\text{gen}}$. 
Therefore, $\mathcal{F}_{\mathcal{G},\text{gen}}$, 
which is the set of labels of all group-covariant generalized-extreme channels, is the set of irreps of $\mathcal{G}$ with dimension not exceeding~$d$. 
Hence, solving Eq.~(\ref{eq:CovariantKraus}) for all $[\Omega]\in\mathcal{F}_{\mathcal{G},\text{gen}}$
ensures that we construct all  group-covariant generalized-extreme channels. Actually an innovation in proposing a successful algorithmic approach for constructing all group-covariant generalized extreme channels, is restriction to irreps $\Omega$  with dimension not exceeding $d$. By testing whether or not~$\mathcal{S}$ in Eq.~(\ref{eq:SAA}) is a set of linearly independent operators, $\Phi_{D^{(1)},D^{(2)},[\Omega]}\in\mathcal{W}_{\mathcal{G},\text{gen}}$ are classified as either  group-covariant extreme or group-covariant quasi-extreme channels. 
 
\subsubsection{Equivalent channels}
\label{subsubsec:EquiChannels}

For a finite discrete group or a compact connected Lie group $\mathcal{G}$ and Hilbert-space dimension~$d$,
the set of group-covariant channels is denoted $\mathcal{W}_{\mathcal{G}}$~(\ref{eq:W_G}), where the union is over all $d$-dimensional representations of $\mathcal{G}$. In this subsubsection, by using the 
equivalence relation between quantum channels, Definition~(\ref{def:channelequivalence}),  we partition the set of group-covariant quantum channels. 
Then we show that,
if an element of an equivalence class is extreme, then all channels in that class are extreme.
Thus, we show that group-covariant channels with respect to equivalent representations of~$\mathcal{G}$ belong to the same equivalence class and if an element of an equivalence class is extreme/quasi-extreme, then all elements of that class are extreme/quasi-extreme.
This observation leads us to construct  group-covariant extreme channels more efficiently. 

Equivalence relation $\sim$ in Eq.~(\ref{eq:equivalence}) partitions
the set of quantum channels
into disjoint equivalence classes $[\Phi]$.
The next lemma shows that, if one of the channels in an equivalence class  $[\Phi]$ is extreme, the other elements of that class are extreme as well. 
\begin{lemma}
\label{lemma:equiextreme}
If channel~$\Phi\in{S_{\Phi_{\text{ext}}}}$, then all channels belonging to $[\Phi]$ are extreme. 
\end{lemma}
\begin{proof}
Let an extreme channel~$\Phi$ be described by a set of Kraus operators $\{A_k\}_{k\in[K]}$.
If the channel $\Phi'$,
described by a set of Kraus operators $\{A'_k\}_{k\in[K]}$, belongs to the same class as~$\Phi$, then unitary operators $U$ and $V$ exist such that Eq.~(\ref{eq:equivalence}) holds. Therefore, 
\begin{equation}
    A'_k=\sum_m W_{km}U^{\dagger}A_mV^{\dagger}
\end{equation}
with $W$ being a $d$-dimensional unitary operator. To show that $\Phi'$ is an extreme channel, we show that operators $\{{A'_k}^{\dagger}A'_\ell\}$ are linearly independent (See Theorem~\ref{theorem:Choi}).
If
\begin{equation}
\sum_{k,\ell}X_{k\ell}{A'_k}^{\dagger}A'_\ell=0,
\end{equation}
then
\begin{equation}
\sum_{m,n}(W^{\dagger}XW)_{mn}A_m^{\dagger}A_n=0
\end{equation}
 is easy to see.
 As $\Phi$ is extreme, according to Theorem~\ref{theorem:Choi}, $X=0.$
Therefore, the set of operators $\{{{A'}_k}^{\dagger}A'_\ell\}$ is linearly independent.
Hence, $\Phi'\in{S_{\Phi_{\text{ext}}}}$.
\end{proof}

Per Definition~\ref{def:channelequivalence}, 
if $D^{(i)}$ and ${D'}^{(i)}$ for $i\in[2]$ are unitarily equivalent representations of the group (\ref{eq:equireps}), 
then Eq.~(\ref{eq:VPhiU}) holds. Therefore, according to Remark~\ref{remark: UniEquivReps},
channels $\Phi_{D^{(1)},D^{(2)},[\Omega]}\in\mathcal{W}_{\mathcal{G},D^{(1)},D^{(2)}}$ and $\Phi'_{D'^{(1)},D'^{(2)},[\Omega]}\in\mathcal{W}_{\mathcal{G},D'^{(1)},D'^{(2)}}$ are equivalent and belong to the same class. We denote the set of all group-covariant channels with respect to representations of $\mathcal{G}$,
which are unitarily equivalent 
to~$D^{(1)}$ and $D^{(2)}$,
by $\mathcal{W}_{\mathcal{G},[D^{(1)}],[D^{(2)]}}$, where labels of distinct classes $[D^{(i)}]$s are inequivalent representations of the group. 
\begin{remark}
\label{remark:equivalentgenextchannels}
According to Lemma \ref{lemma:equiextreme}, if 
$\Phi_{D^{(1)},D^{(2)},[\Omega]}\in\mathcal{W}_{\mathcal{G},[D^{(1)}],[D^{(2)}]}$ is extreme, then all other channels in $\mathcal{W}_{\mathcal{G},[D^{(1)}],[D^{(2)}]}$ are extreme.
Hence, for group-covariant generalized-extreme channels,
that is for all $\Phi_{D^{(1)},D^{(2)},[\Omega]}$ with $[\Omega]\in\mathcal{F}_{\mathcal{G},\text{gen}}$ (defined in \S\ref{subsubsec:constructinggegcchannels}), if 
$\Phi_{D^{(1)},D^{(2)},[\Omega]}\in\mathcal{W}_{\mathcal{G},[D^{(1)}],[D^{(2)}]}$ is extreme/quasi-extreme, then all other channels in $\mathcal{W}_{\mathcal{G},[D^{(1)}],[D^{(2)}]}$ are extreme/quasi-extreme. Therefore,
to construct $\mathcal{W}_{\mathcal{G}}$,
instead of constructing all elements of $\mathcal{W}_{\mathcal{G},D^{(1)},D^{(2)}}$ and then taking the union over $d$-dimensional representation of $\mathcal{G}$,
we just construct one channel~$\Phi_{D^{(1)},D^{(2)},[\Omega]}$ in each $\mathcal{W}_{\mathcal{G},[D^{(1)}],[D^{(2)}]}$ 
with $[\Omega]\in\mathcal{F}_{\mathcal{G},\text{gen}}$. The other extreme/quasi-extreme channels can be constructed from this representative channel by unitary conjugations before and after the action of the channel.
\end{remark}
%**********************************************
%-----------------------------------------
\subsection{Pseudocoding the construction of generalized extreme group-covariant channels:}
\label{subsec:pseudocoding}
%\paragraph{About:}
In this subsection we present our pseudocode results for constructing descriptions of group-covariant generalized-extreme channels.
Typically construction of channel descriptions are accomplished mathematically but of course not successfully in general due to difficult in solving such generic problems.
Our approach is not to write the mathematical solution in general but rather the algorithm that generates the solution on an appropriate autonomous logical machine with sufficient resources.
Our approach synthesizes techniques from computing and representation theory and thus involves knowledge from both fields. 
Here, we express our algorithmic approach for solving Problem \ref{prob:discrete}, Problem  \ref{prob:Lie} and Problem \ref{prob:decision} as pseudocodes.
For more details on pseudocode, notation and library used in the pseudocode see \S\ref{subsec:Pseudocode}, \S\ref{subsubsec:datatypes} and \S\ref{subsubsec:library}.
%------------------------algorithm 1 --------------------------------
\begin{algorithm}[H]
\caption{Generate all group-covariant generalized-extreme channels up to unitary equivalence for a finite discrete group or a compact connected Lie group}
\label{alg:Krausoperatorsdiscgroup}
\begin{algorithmic}[1]
\Require 
\Statex \texttt{character} \textsc{gName}
\Comment Name of finite discrete or compact connected Lie group~$\mathcal{G}$
\Statex \texttt{binary} \textsc{gType}\Comment{$0$ for finite discrete group; $1$ for compact connected Lie group}
\Statex \texttt{natural} \textsc{hDim}
\Comment{Hilbert-space dimension}
\Ensure
\Statex \texttt{symjagged}[ ][ ][\textsc{hDim}][\textsc{hDim}] \textsc{kraus}\Comment {Set of Kraus operators with dimension \textsc{hDim}}
\Statex \texttt{natural} \textsc{count}\Comment{Number of generalized extreme \textsc{gName}-covariant channels}
\Procedure{\textsc{KrOpDiscG}}{\textsc{gName},\textsc{gType},\textsc{hDim}}
\State IMPORT \textsc{solveChannel}, \textsc{delete}, \textsc{dag}, \textsc{id}, \textsc{symSolve}, \textsc{props}, \textsc{transp}, \textsc{vecToMatr}, \textsc{zero}\Comment{Import from libraries~\S\ref{subsubsec:blibrary},~\S\ref{subsubsec:rlibrary}}
\State \texttt{posinteger} \textsc{numIrrep}
\State \texttt{posinteger} \textsc{numRep}
\State \texttt{posinteger}  \textsc{numGen}
\State \texttt{posinteger}[ ] \textsc{dim}
\State \texttt{symbol}[ ][ ][\textsc{hDim}][\textsc{hDim}] \textsc{rep}
\State \texttt{symjagged}[ ][ ][ ][ ] \textsc{irrep}
\State \texttt{symjagged}[ ][ ][\textsc{hDim}][\textsc{hDim}] \textsc{kraus}
\State \texttt{\texttt{natural}} \textsc{count}\Comment{Counts the number of generalized extreme \textsc{gName}-covariant channels}
\State (\textsc{numIrrep},\textsc{numRep},\textsc{numGen},\textsc{dim},\textsc{irrep},\textsc{rep})$\gets$ \textsc{props}(\textsc{gName},\textsc{gType},\textsc{hDim})
\State \texttt{symbol}[\textsc{numRep}][\textsc{numRep}][\textsc{numIrrep}][\textsc{numGen}][ ][ ]\textsc{b}\Comment{Coefficient matrix in Eq.~(\ref{eq:LinearizedCovariantKrausDis})/(\ref{eq:LinearizedCovariantKrausLie}) (\textsc{gType}=0/1)}
%\State \texttt{natural}[3] \textsc{gProp} $\gets$ \textsc{gProperties}(\textsc{gName})
\State \textsc{count}$\gets 1$
\For{\textsc{m}$\gets0$ to $\textsc{numIrrep}-1$}

\For{\textsc{k}$\gets0$ to $\textsc{numRep}-1$}
   \For{\textsc{l}$\gets0$ to $\textsc{numRep}-1$}
   \If{\textsc{gType}$=0$}
   \For{\textsc{i}$\gets0$ to $\textsc{numGen}-1$}
      \State\textsc{b}[\textsc{count}][\textsc{i}]$\gets\bigoplus_{j=0}^{\textsc{dim}[\textsc{m}]-1}\left(\textsc{dag}(\textsc{gRep}[\textsc{k}][\textsc{i}])
      \otimes\textsc{transp}(\textsc{gRep}[\textsc{l}][\textsc{i}])\right)-\textsc{gIrrep}[\textsc{m}][\textsc{i}]\otimes\textsc{id}(\textsc{hDim}^2)$
   \EndFor
   \Else
      \For{\textsc{i}$\gets0$ to $\textsc{numGen}-1$}
      \State\textsc{b}[\textsc{count}][\textsc{i}]$\gets\bigoplus_{j=0}^{\textsc{dim}[\textsc{m}]-1}\left(\textsc{gRep}[\textsc{k}][\textsc{i}]\otimes\textsc{id}(\textsc{hDim})-\textsc{id}(\textsc{hDim})\otimes
   \textsc{dag}(\textsc{gRep}[\textsc{l}][\textsc{i}])\right)-\textsc{omega}[\textsc{m}][\textsc{i}]\otimes\textsc{id}(\textsc{hDim}^2)$
   \EndFor
   \EndIf
   \State \texttt{symbol}[$\textsc{dim}[\textsc{m}]*\textsc{hDim}^2$] \textsc{x}
  \State \textsc{x}$\gets$ \textsc{symSolve}(\textsc{b}[\textsc{count}])
   \If{\textsc{x}$\neq$ \textsc{zero}($\textsc{dim}[\textsc{m}]*\textsc{hDim}^2$)}\Comment{Tests if the covariant CP exists}
    \State \textsc{kraus}[\textsc{count}]$\gets$\textsc{vecToMatr}(\textsc{x},\textsc{dim}[\textsc{m}])
    \State\textsc{kraus}[\textsc{count}]$\gets$\textsc{solveChannel}(\textsc{kraus}[\textsc{count}])
    \If{\textsc{kraus}$\neq$\textsc{zero}(\textsc{dim}(\textsc{m}),\textsc{hDim},\textsc{Hdim})}\Comment{Tests if the covariant channel exists}
    \State \textsc{count}$\gets\textsc{count}+1$
    \EndIf
    \State \textsc{delete}(\textsc{x})\Comment{Delete \textsc{x} and clear space}
    \EndIf
   \EndFor
\EndFor
\EndFor
\State RETURN \textsc{count}
\State RETURN \textsc{kraus}
\EndProcedure
\end{algorithmic}
\end{algorithm}
%------------------------------------------------------
\begin{algorithm}[H]
\caption{Test whether a given channel is extreme or not}
\label{alg:TestExtremality}
\begin{algorithmic}[1]
\Require
\texttt{symbol}[\textsc{num}][\textsc{hDim}][\textsc{hDim}] \textsc{kraus}\Comment{Kraus operators of the channel}
  \Ensure
\Statex \texttt{binary} \textsc{isExtreme}\Comment{ \textsc{isExtreme} is~0 for an extreme channel and~1 for a quasi-extreme channel}
\Procedure{testExt}{\textsc{kraus}}
\State IMPORT \textsc{dag}, \textsc{symSolve}, 
\textsc{zero} \Comment{Import from libraries~\S\ref{subsubsec:blibrary} and~\S\ref{subsubsec:rlibrary}}
\State \texttt{binary} \textsc{isExtreme}
\State \texttt{natural} \textsc{n}\Comment{Counts loops}
\State \texttt{symbol}[1][\textsc{$\textsc{hDim}^2$}][$\textsc{num}^2$] \textsc{matCoeff}\Comment{Matrix of coefficients}
\State \texttt{symbol}[$\textsc{num}^2$] \textsc{solusion}\Comment{Solution of system of linear equations}
\State \texttt{symbol}[$\textsc{num}^2$][\textsc{hDim}][\textsc{hDim}] \textsc{scratchPad}\Comment{Scratch space for calculating}
\For{\textsc{i}$\gets0$ to $\textsc{num}-1$}
    \For{\textsc{j}$\gets0$ to $\textsc{num}-1$}
      \State\textsc{scratchPad}[\textsc{n}]$\gets$\textsc{dag}(\textsc{kraus}[\textsc{i}])$\times$\textsc{kraus}[\textsc{j}]
    \EndFor
    \State \textsc{n}$\gets\textsc{n}+1$
\EndFor
\State \textsc{n}$\gets0$
\For{\textsc{i}$\gets0$ to $\textsc{hDim}^2-1$}
     \For{\textsc{j}$\gets\textsc{i}+1$} to $\textsc{hDim}^2-1$
         \For{\textsc{k}$\gets0$ to $\textsc{num}^2-1$}
            \State \textsc{matCoeff}[0][\textsc{n}][\textsc{k}]$\gets$\textsc{scratchPad}[\textsc{k}][\textsc{i}][\textsc{j}]
         \EndFor
         \State\textsc{n}$\gets\textsc{n}+1$
     \EndFor
\EndFor
\State \textsc{x}$\gets$\textsc{symSolve}(\textsc{matCoeff})
\If{$\textsc{x}\neq\textsc{zero}\left(\textsc{num}^2\right)$} 
   \State\textsc{isExtreme}$\gets1$
   \Comment{\textsc{isExtreme} indicates quasiextreme channel; else extreme (0)}
\EndIf
\State RETURN \textsc{isExtreme}
\EndProcedure
\end{algorithmic}
\end{algorithm}
%\end{landscape}
%----------------------------------------------------
\subsection{Explicit examples:}
\label{subsec:explicitexamples}
In this subsection, we present explicit examples for constructing group-covariant extreme or quasi-extreme channels for $Z_2$, $S_3$, $A_4$ and $D_5$ as finite discrete groups~\cite{CoxeterMoser}
and $SO(3)$ and $SU(2)$ for compact connected Lie groups~\cite{BarutRaczka}.
We recall from~\S\ref{subsec:formalapproach}
and~\S\ref{subsec:mathematical} that for any given group, group-covariant generalized-extreme channels are denoted by
$\Phi_{[D^{(1)]},[D^{(2)}],[\Omega]}$, where~$D^{(1)}$ and~$D^{(2)}$ are representative of distinct classes of $d$-dimensional representations of the group and~$\Omega$ is representative of class of irreducible representations of the group with dimension less than or equal to~$d$. 
In making choices for~$\Omega$,
we use the theorem that says that the number of inequivalent irreps for a finite discrete group equals the number of conjugacy classes~\cite{Herstein}.
\subsubsection{$Z_2$-covariant extreme qubit channels}
\label{subsec:Z2}
%\paragraph{About}
Our first example is the group~$Z_2$,
which is abelian,
hence has just one-dimensional irreps.
As the number of Kraus operators equals the dimension of the irrep,
a $Z_2$-covariant channel has just one Kraus operator and,
due to the trace-preserving constraint~(\ref{eq:TP}),
this channel is strictly a unitary channel. Therefore, the constructed channel is certainly extreme~\cite{Cho75}.

The $Z_2$ group is explicitly
\begin{equation}
    Z_2:=\Braket{g}/\{g^2
        =e\}=\{e,g\}
\end{equation}
with~$e$ being the identity element of the group
and the notation~$\Braket{}$ denoting the generating set~(\ref{eq:Sgenerators}).
This group has two inequivalent irreducible representations, namely,
$\Omega_\pm^{(1)}(g)=\pm1$,
which are one-dimensional. 
Thus, we have two options for~$\Omega$ in Eq.~(\ref{eq:LinearizedCovariantKrausDis}).
For a single-qubit $Z_2$-covariant channel,
$D^{(1)}$ and~$D^{(2)}$ are two-dimensional representations of the group,
and uncountably many two-dimensional representations for each of~$D^{(1)}$ and~$D^{(2)}$ are allowed.
According to Lemma~\ref{lemma:equiextreme} and Remark~\ref{remark:equivalentgenextchannels},
we just consider inequivalent two-dimensional representations.
These two cases are constructed as direct sums of inequivalent irreps of $Z_2$, i.e., $\Omega^{(1)}_+\oplus\Omega_\pm^{(1)}$,
Therefore, we have two cases for $D^{(1)}$, two cases for~$D^{(2)}$ and two cases for~$\Omega$ so we have eight cases to consider.
%Below we summarize the results.

%First we consider the case
%    \begin{equation}
%        D^{(1)}(g)=D^{(2)}(g)
%        =\operatorname{diag}(1, -1), \; \Omega^{(1)}(g)=1.
%    \end{equation}
%This channel has a single unitary Kraus operator

Now we proceed to investigate the first of these eight cases.
The unitary Kraus operator representing $\Phi_{\left[\Omega_+^{(1)}\oplus\Omega_-^{(1)}\right],\left[\Omega_+^{(1)}\oplus\Omega_-^{(1)}\right],\left[\Omega_+^{(1)}\right]}$
is represented by~(\ref{eq:Krausrep})
    \begin{equation}
     A=\operatorname{diag}(1, \text{e}^{\text{i}\alpha}),\; \alpha\in\mathbb{R}.
\end{equation}
This matrix~$A$ is evidently unitary and satisfies the trace-preserving constraint~(\ref{eq:TP}).
From the trivial linear independence of the single-element set
$\mathcal{S}=\{A^\dagger A\}$~(\ref{eq:SAA}),
$\Phi_{\left[\Omega_+^{(1)}\oplus\Omega_-^{(1)}\right],\left[\Omega_+^{(1)}\oplus\Omega_-^{(1)}\right],\left[\Omega_+^{(1)}\right]}$ must be extreme.

Now we move to the second case.
The unitary Kraus operator of $\Phi_{\left[\Omega_+^{(1)}\oplus\Omega_-^{(1)}\right],\left[\Omega_+^{(1)}\oplus\Omega_-^{(1)}\right],\left[\Omega_-^{(1)}\right]}$ is
    \begin{equation}
    A=\begin{pmatrix}
    0&\text{e}^{\text{i}\beta}\\
        1&0
    \end{pmatrix},\; \beta\in\mathbb{R}.
\end{equation}      
Evidently, this matrix is unitary, thus is trace preserving and extreme as expected.

Now we consider the case of $\Phi_{\left[\Omega_+^{(1)}\oplus\Omega_+^{(1)}\right],\left[\Omega_+^{(1)}\oplus\Omega_+^{(1)}\right],\left[\Omega_+^{(1)}\right]}$.
In this case,
following from Remark~\ref{remark:trivialcase},
the group-covariant property does not constrain the single Kraus operator.
Hence, the only constraint is the trace-preserving condition,
which results in the map $\Phi_{\left[\Omega_+^{(1)}\oplus\Omega_+^{(1)}\right],\left[\Omega_+^{(1)}\oplus\Omega_+^{(1)}\right],\left[\Omega_+^{(1)}\right]}$ being a general unitary evolution. 
In the fourth case, $\Phi_{\left[\Omega_+^{(1)}\oplus\Omega_+^{(1)}\right],\left[\Omega_+^{(1)}\oplus\Omega_+^{(1)}\right],\left[\Omega_-^{(1)}\right]}$,
the Kraus operator is zero,
which means that such a group-covariant channel does not exist. 

Finally, we consider the last four cases, which are 
\begin{equation}
\label{eq:last4Z2cases}
    \Phi_{\left[\Omega_+^{(1)}\oplus\Omega_-^{(1)}\right],\left[\Omega_+^{(1)}\oplus\Omega_+^{(1)}\right],\left[\Omega_+^{(1)}\right]},\;
    \Phi_{\left[\Omega_+^{(1)}\oplus\Omega_-^{(1)}\right],\left[\Omega_+^{(1)}\oplus\Omega_+^{(1)}\right],\left[\Omega_-^{(1)}\right]},\;
    \Phi_{\left[\Omega_+^{(1)}\oplus\Omega_+^{(1)}\right],\left[\Omega_+^{(1)}\oplus\Omega_-^{(1)}\right],\left[\Omega_+^{(1)}\right]},\;
    \Phi_{\left[\Omega_+^{(1)}\oplus\Omega_+^{(1)}\right],\left[\Omega_+^{(1)}\oplus\Omega_-^{(1)}\right],\left[\Omega_-^{(1)}\right]}.
\end{equation}
All four of these maps~(\ref{eq:last4Z2cases})
fail to satisfy the trace-preserving constraint.
Thus, these cases do not occur as outputs from the algorithm. 

This simple example of $Z_2$-covariant channels illustrates how our method leads to successful results, 
which are already known for the case of qubit channels~\cite{FA98,RSW02}.
To go beyond unitary group covariant channels, we need to consider cases of non-abelian groups.
Hence, in~\S\ref{subsubsec:S3},
we study the permutation group.
\subsubsection{$S_3$-covariant qubit and qutrit  extreme and quasi-extreme}
\label{subsubsec:S3}
In this subsubsection we consider $S_3$-covariant channels with~$S_3$ the symmetric group, which is a set of all permutations of three objects.
For all abelian groups,
we always obtain unitary channels.
Hence,
to obtain nonunitary channels,
we begin by choosing~$S_3$, which is the smallest nonabelian group.

The order of the group~$S_3$, is six that is $|S_3|=6$ with $|G|$ denoting the order of the group $G$. It has two generators denoted by $\sigma_1$ and $\sigma_2$
such that
\begin{equation}
    S_3=\Braket{\sigma_1,\sigma_2}/\{\sigma_1^2=\sigma_2^2=e, \sigma_1\sigma_2\sigma_1=\sigma_2\sigma_1\sigma_2\}.
\end{equation}
$S_3$ has three inequivalent irreps,
specifically, two one-dimensional irreps
\begin{equation}
\label{eq:Irrep1S3}
\Omega^{(1)}_+(\sigma_1)=1=\Omega^{(1)}_+(\sigma_2),\;
\Omega_-^{(1)}(\sigma_1)=1=-\Omega^{(1)}_-(\sigma_2),
\end{equation}
and one two-dimensional irrep
\begin{equation}
\label{eq:Irrep2S3}
    \Omega^{(2)}(\sigma_1)
        =\operatorname{diag}(1,-1),\;
    \Omega^{(2)}(\sigma_2)
        =\frac12
        \begin{pmatrix}
        -1&\sqrt{3}\text{e}^{\text{i}\phi}\\
        \sqrt{3}\text{e}^{-\text{i}\phi}&1\\
        \end{pmatrix}, \phi\in[0,2\pi].
\end{equation}
Representations of~$S_3$ with dimension larger than two are indeed reducible and can be constructed as unitary conjugations of direct sums of its irreps. 
\paragraph{Qubit channels:}
\label{par:S3qubit}
For qubit channels,
$d=2$,
$D^{(1)}$ and~$D^{(2)}$ are two-dimensional representations of $S_3$.
According to Lemma~{\ref{lemma:equiextreme}} and Remark~\ref{remark:equivalentgenextchannels}, among uncountable number of two-dimensional representations for $S_3$, it is sufficient to consider just inequivalent two-dimensional reps of~$S_3$ to solve Eq.~(\ref{eq:LinearizedCovariantKrausDis}). In fact there are just four inequivalent two-dimensional representations for~$S_3$ given by \begin{equation}
    \Omega^{(2)}, \Omega_+^{(1)}\oplus\Omega_+^{(1)},\; \Omega_+^{(1)}\oplus\Omega_-^{(1)},\; \Omega_-^{(1)}\oplus\Omega_-^{(1)}.
\end{equation}
For~$\Omega$,
we first note that one-dimensional irreps of the groups
cannot lead to non-unitary channels. Hence, to go beyond unitary channels, we have just one option for~$\Omega$, that is $\Omega^{(2)}$~(\ref{eq:Irrep2S3}). 
Thus,
given that,
for each~$D^{(1)}$ and~$D^{(2)}$,
we have four options, 
the total number of cases to be solved for non-unitary $S_3$-covariant channels is sixteen. 

Among these sixteen cases, here we just report one case, that is  $\Phi_{[\Omega^{(2)}],[\Omega^{(2)}],[\Omega^{(2)}]}$. For this case the only solution to set of Eqs.~(\ref{eq:LinearizedCovariantKrausDis}), is zero. Thus, $\Phi_{[\Omega^{(2)}],[\Omega^{(2)}],[\Omega^{(2)}]}$ does not exist. 
Although we have only treated one of sixteen cases here,
this case illustrates the feasibility of our method and the qubit case is fully understood~\cite{RSW02} so we move on to the unexplored $d=3$ cases instead.

\paragraph{Qutrit channels:}
\label{par:S3qutrit}
For qutrit channels, $d=3$,
and the dimension of all irreps of~$S_3$ satisfies $d<3$;
hence, the necessary condition in Lemma~\ref{lemma:NecCondOmega} is satisfied for~$\Omega$ being any~$S_3$ irrep.
Among these three choices for~$\Omega$, to go beyond unitary channels,
we choose $\Omega=\Omega^{(2)}$~(\ref{eq:Irrep2S3})
as other irreps of $S_3$, namely $\Omega_\pm^{(1)}$,
are one-dimensional~(\ref{eq:Irrep1S3}) and label channels that are unitary.
According to Remark~\ref{remark:equivalentgenextchannels},
candidates for~$D^{(1)}$ and~$D^{(2)}$ are all three-dimensional inequivalent representations of~$S_3$ which we construct by direct sum of irreps of $S_3$.
Direct sums of one-dimensional irreps yield four inequivalent irreps,
namely
\begin{equation}
\label{eq:4ineqirreps}
\Omega_+^{(1)}\oplus\Omega_+^{(1)}\oplus\Omega_+^{(1)},\; 
\Omega_+^{(1)}\oplus\Omega_+^{(1)}\oplus\Omega_-^{(1)},\;
\Omega_+^{(1)}\oplus\Omega_-^{(1)}\oplus\Omega_-^{(1)},\;
\Omega_-^{(1)}\oplus\Omega_-^{(1)}\oplus\Omega_-^{(1)},
\end{equation}
which are direct sums of one-dimensional and two-dimensional irreps 
and also yield two more inequivalent cases
\begin{equation}
\Omega_\pm^{(1)}\oplus\Omega^{(2)}.
\end{equation}
Hence,
for each~$D^{(1)}$ and~$D^{(2)}$,
we have six candidates.

In total for strictly non-unitary $S_3$-covariant channels,
36 cases need to be solved.
Among all these 36 cases,
we just solve one case as an illustration.
One case suffices to confirm our method before proceeding to test our method on different groups.

We now solve Eqs.~(\ref{eq:LinearizedCovariantKrausDis}) for Kraus operators with
\begin{equation}
    D^{(1)}=D^{(2)}=\Omega_+^{(1)}\oplus\Omega^{(2)},\;
    \Omega=\Omega^{(2)},
\end{equation}
which yields the Kraus-operator family
\begin{equation}
\label{eq:KrausS3d3}
    A_1=\begin{pmatrix}
                0& \alpha&0 \\
                 \beta& \gamma&0\\
                 0&0&-\gamma
            \end{pmatrix},\;
            A_2=\begin{pmatrix}
                0&0&\alpha \\
                0&0&-\gamma\\
                \beta \text{e}^{-2\text{i}\phi}&-\gamma \text{e}^{-2\text{i}\phi}&0\\
            \end{pmatrix},\;
    \alpha,\beta,\gamma\in\mathbb{C},
\end{equation}
which represent an $S_3$-covariant completely positive map. 
By imposing the trace-preserving constraint~(\ref{eq:TP}),
we obtain the relations
\begin{equation}
\label{eq:systemS3d3}
    2|\beta|^2=1,\;
    |\alpha|^2+2|\gamma|^2=1.
\end{equation}
between the parameters in Eq.~(\ref{eq:KrausS3d3}).
Hence, Kraus operators of $S_3$-covariant channel with respect to, $\Omega_+^{(1)}\oplus\Omega^{(2)}$ and labelled by $\Omega^{(2)}$, that is
$\Phi_{[\Omega_+^{(1)}\oplus\Omega^{(2)}],[\Omega_+^{(1)}\oplus\Omega^{(2)}],[\Omega^{(2)}]}$
, are given in Eq.~(\ref{eq:KrausS3d3}) with the constraint in Eq.~(\ref{eq:systemS3d3}).

To see which parameter values correspond to extreme channels,
we construct the set~$\mathcal{S}$~(\ref{eq:SAA}) for Kraus operators~(\ref{eq:KrausS3d3}).
The obtained set is a set of linearly independent operators for all values of parameters in Eq.~(\ref{eq:KrausS3d3}) except for the instance
\begin{equation}
\label{eq:S3d3GE}
        |\alpha|^2=\frac12,\; |\gamma|^2=\frac14.
\end{equation}
That is,
the set of Kraus operators~(\ref{eq:KrausS3d3}) with constraint~(\ref{eq:systemS3d3}) represents $S_3$-covariant extreme channels for the entire range of parameters except at the point~(\ref{eq:S3d3GE}), which represents a $S_3$-covariant quasi-exteme channel. 
\subsubsection{$A_4$-covariant extreme qutrit channel:}
\label{subsubsec:A4}
%\paragraph{About:}
Our next example is for the alternating group $A_4$,
which is a normal subgroup of $S_4$,
and~$A_4$ consists of all even permutations of a four-object set.
Although~$S_4$ is the natural next case after~$S_3$ in~\S\ref{subsubsec:S3},
$S_4$ has three generates so we restrict to the more manageable case of~$A_4$ here for our generalization.
The group~$A_4$ is nonabelian and of order~$12$,
and~$A_4$ has two generators,
which we denote by $g_1$ and $g_2$.
Furthermore $A_4$ has four inequivalent irreducible representations.
Three of these representations are one-dimensional and are denoted by~$\Omega_{+,\circ,-}^{(1)}$. 
The fourth case is the three-dimensional irreducible representation
\begin{equation}
\label{eq:A4Omegag12}
    \Omega^{(3)}(g_1)=\operatorname{diag}(1,\omega,
    \omega^2),\;
    \Omega^{(3)}(g_2)
    =-\frac13\begin{pmatrix}
                1 & -2\omega^2 & 2\omega\\
                -2 &\omega^2 & 2\omega\\
                2& 2\omega^2&\omega
                \end{pmatrix},
    \omega:=\text{e}^{\frac{2i\pi}{3}}.
\end{equation}

For a qutrit channel $d=3$,
which determines the dimension of both~$D^{(1)}$ and $D^{(2)}$.
Based on Remark~\ref{remark: UniEquivReps} for~$D^{(1)}$ and $D^{(2)}$,
we only consider inequivalent three-dimensional cases from amongst the uncountable number of available three-dimensional representations of $A_4$.
Hence, for each~$D^{(1)}$ and $D^{(2)}$,
11 candidates
\begin{align}
    &\Omega_+^{(1)}\oplus\Omega_+^{(1)}\oplus\Omega_+^{(1)},\;
    \Omega_\circ^{(1)}\oplus\Omega_\circ^{(1)}\oplus\Omega_\circ^{(1)},\;
    \Omega_-^{(1)}\oplus\Omega_-^{(1)}\oplus\Omega_-^{(1)},\;\cr
    &\Omega_+^{(1)}\oplus\Omega_+^{(1)}\oplus\Omega_\circ^{(1)},\;
    \Omega_+^{(1)}\oplus\Omega_\circ^{(1)}\oplus\Omega_\circ^{(1)},\;
    \Omega_+^{(1)}\oplus\Omega_+^{(1)}\oplus\Omega_-^{(1)},\;\cr
    &\Omega_+^{(1)}\oplus\Omega_-^{(1)}\oplus\Omega_-^{(1)},\;
    \Omega_\circ^{(1)}\oplus\Omega_\circ^{(1)}\oplus\Omega_-^{(1)},\;
    \Omega_\circ^{(1)}\oplus\Omega_-^{(1)}\oplus\Omega_-^{(1)},\;\cr
    &\Omega_+^{(1)}\oplus\Omega_\circ^{(1)}\oplus\Omega_-^{(1)},\; \Omega^{(3)}
\end{align}
exist.
To select a proper representation for $\Omega$, we note that its dimension should be less than or equal to three.
Therefore, according to Lemma~\ref{lemma:NecCondOmega}, all irreps of $A_4$ satisfy the necessary condition for labeling an $A_4$-covariant generalized-extreme channel.
The three one-dimensional irreps of $A_4$ label unitary channels, if they exist. Hence, the only candidate for~$\Omega$ to label a non-unitary $A_4$-covariant generalized-extreme channel is $\Omega^{(3)}$.
Thus, for non-unitary $A_4$-covariant generalized-extreme channel in total,
there are 121 cases to study due to 11 choices we have for each $D^{(1)}$ and $D^{(2)}$. Among these we focus on one of them just to illustrate how our method works.

Solving Eq.~(\ref{eq:LinearizedCovariantKrausDis}) for channel $\Phi_{[\Omega^{(3)}],[\Omega^{(3)}],[\Omega^{(3)}]}$
yields Kraus operators
\begin{equation}
\label{eq:A4constantmatrices}
    A_1=\frac{1}{\sqrt{2}}\operatorname{diag}(0, 1, -1),\;
    A_2=\frac{1}{\sqrt{2}}\begin{pmatrix}
    0& -1 &0\\
    0&0&0\\
    \omega&0&0\\
    \end{pmatrix},\;
    A_3=\frac{1}{\sqrt{2}}\begin{pmatrix}
      0&0&1\\
      -\omega&0&0\\
      0&0&0
    \end{pmatrix},
\end{equation}
which is a set of constant-valued matrices rather than a parameter family~(\ref{eq:KrausS3d3}).
These Kraus operators satisfy the trace-preserving condition (\ref{eq:TP}).
Also the set of operators $\mathcal{S}=\{A_i^\dagger A_j\}$ with $i,j=1,2,3$ is linearly independent.
Hence, $\Phi_{[\Omega^{(3)}],[\Omega^{(3)}],[\Omega^{(3)}]}$ is an extreme quantum channel.
\subsubsection{$D_5$-covariant extreme qutrit channel}
\label{subsubsec:D_5}
%\paragraph{About:}
For $S_3$ in~\S\ref{subsubsec:S3},
we obtain Kraus operators~(\ref{eq:KrausS3d3}) of a family of extreme channels.
Instead of extending to $S_4$, for tractability reasons we extending to $A_4$ which is a restriction of $S_4$ to even permutations. That yields only constant matrices~(\ref{eq:A4constantmatrices})
as the solution.
Now we consider another `small' group,
dihedral group~$D_5$,
which is the group of symmetries for a regular pentagon. 
By studying~$D_5$,
we can see whether the $S_3$-covariant and $A_4$-covariant extreme channels are also extreme channels for a $D_5$-covariant channel or not
despite $D_5\nleq S_3$
and $D_5\nleq A_4$.

The group
\begin{equation}
\label{eq:D5}
    D_5=\Braket{g_1,g_2}/\{g_1^2=g_2^5=(g_1g_2)^2=e\}.
\end{equation}
is of order ten with two generators, namely $g_1$ and $g_2$,
which generate reflection and rotation, respectively.
$D_5$ has four inequivalent irreps.
Two of the inequivalent irreps are one-dimensional,
\begin{equation}
    \Omega_+^{(1)}(g_1)=1=\Omega_+^{(1)}(g_2)=1,\;
    \Omega_-^{(1)}(g_1)=-1=-\Omega_-^{(1)}(g_2),
\end{equation}
and the other two are two-dimensional,
\begin{equation}
    \Omega_+^{(2)}(g_1)=\begin{pmatrix}
    1 &0\\
    0& -1 \\
    \end{pmatrix},\;
    \Omega_+^{(2)}(g_2)
        =\begin{pmatrix}
    \cos\omega&-\sin\omega\\
    \sin\omega& \cos\omega\\
    \end{pmatrix},\;
    \Omega_-^{(2)}(g_1)=\begin{pmatrix}
    1 &0\\
    0& -1 \\
    \end{pmatrix},\;
    \Omega^{(2)}_-(g_2)=\begin{pmatrix}
    \cos2\omega&-\sin2\omega\\
    \sin2\omega& \cos2\omega\\
    \end{pmatrix}.
\end{equation}
for $\omega:=\text{e}^{\frac{2i\pi}{5}}$.
Indeed representations of $D_5$ in dimensions other than one and two are reducible and are constructed by unitary conjugation of direct sums of its irreps. 

In this example we are interested in qutrit $D_5$-covariant generalized-extreme channels.
Hence, $d=3$,
which determines the dimension of~$D^{(1)}$ and $D^{(2)}$.
According to Remark~(\ref{remark:equivalentgenextchannels}),
among the uncountable number of three-dimensional representations of~$D_5$, we only consider inequivalent three-dimensional representations for each~$D^{(1)}$ and $D^{(2)}$.
These eight cases are
\begin{equation}
\label{eq:3dRepsD5}
    \Omega_+^{(1)}\oplus\Omega_+^{(1)}\oplus\Omega_+^{(1)},\;
    \Omega_+^{(1)}\oplus\Omega_+^{(1)}\oplus\Omega_-^{(1)},\;
    \Omega_+^{(1)}\oplus\Omega_-^{(1)}\oplus\Omega_-^{(1)},\;
    \Omega_-^{(1)}\oplus\Omega_-^{(1)}\oplus\Omega_-^{(1)},\;
    \Omega_+^{(1)}\oplus\Omega_+^{(2)},\;
    \Omega_+^{(1)}\oplus\Omega_-^{(2)},\;
    \Omega_-^{(1)}\oplus\Omega_+^{(2)},\;
    \Omega_-^{(1)}\oplus\Omega_-^{(2)}.
\end{equation}
As the dimension of all irreps of $D_5$ satisfies the necessary condition in Lemma \ref{lemma:NecCondOmega}, they are all acceptable candidates to label $D_5$-covariant generalized-extreme channels.
However, only irreps of dimension greater than one can label non-unitary maps.
Hence, for non-unitary $D_5$-covariant generalized-extreme channels,
the two candidate for $\Omega$
are~$\Omega_+^{(2)}$ and~$\Omega_-^{(2)}$. Therefore
for non-unitary $D_5$-covariant generalized-extreme channels, we have 128 cases to study. 
\begin{remark}
The number of non-unitary  group-covariant generalized-extreme channel candidates for group~$D_5$ exceeds the number for~$A_4$~\S\ref{subsubsec:A4}
despite $|A_4|>|D_5|$
and despite the number of inequivalent irreps for $A_4$ equals the number of irreps for~$D_5$.
\end{remark}
Among all these cases we just consider two cases as an example,
namely,
\begin{equation}
    \Phi_{[\Omega_+^{(1)}\oplus\Omega_+^{(2)}],[\Omega_+^{(1)}\oplus\Omega_+^{(2)}],[\Omega_+^{(2)}]},\; \Phi_{[\Omega_+^{(1)}\oplus\Omega_-^{(2)}],[\Omega_+^{(1)}\oplus\Omega_-^{(2)}],[\Omega_-^{(2)}]}.
\end{equation}
Solving Eq.~(\ref{eq:LinearizedCovariantKrausDis}) for these instances, yields the same Kraus operators 
\begin{equation}
\label{eq:KrausD5d3}
    A_1=\frac{1}{\sqrt{2}}\begin{pmatrix}
    0& \sqrt{2}&0\\
    1 &0&0\\
    0&0&0\\
    \end{pmatrix}\;
    A_2=\frac{1}{\sqrt{2}}\begin{pmatrix}
    0&0& \sqrt{2}\\
    0&0&0\\
    1 &0&0
    \end{pmatrix}
\end{equation}
for both cases.
Although $D_5\nleq S_3$, 
by comparing Eqs.~(\ref{eq:KrausD5d3}) and (\ref{eq:KrausS3d3})
we see that the Kraus operators 
for the $D_5$-covariant channels discussed in this example
are a special case of the family of $S_3$-covariant extreme channels with Kraus operators given in Eq.~(\ref{eq:KrausS3d3}). 
\subsubsection{$SO(3)$-covariant channels:}
\label{subsubsec:SO3}
Thus far, 
examples of group-covariant extreme channels have been studied for finite discrete groups.
In this example, we focus on SO(3)-group covariant channels
with SO(3)
generated by the Lie algebra~$\mathfrak{so}(3)$.
The generating set for~$\mathfrak{so}(3)$ 
comprises the ladder operators~$L_\pm$ and the Cartan operator $L_z$,
and~$\{\ket{l,m}\}$
are the weight states labelled by Casimir and Cartan operators,
$L^2$ and~$L_z$,
respectively \cite{BarutRaczka}.
The Casimir invariant has
spectrum~$\{l(l+1);l\in\mathbb{N}\}$ for~$\mathbb{N}$ denoting the set of natural number discussed in~\S\ref{subsubsec:datatypes}.
Thus, unitary irreps of~$\mathfrak{so}(3)$ have odd dimensions 
$\{2l+1;l\in\mathbb{N}\}$ 
and are given by 
\begin{equation}
\label{eq:OmegaSO(3)}
    \Omega^{(2l+1)}(L_\pm)=\sum_{m=-l}^lC^{\pm}_{lm}\ket{l,m\pm1}\bra{l,m},\;
    \Omega^{(2l+1)}(L_z)=\sum_{m=-l}^l m\ket{l,m}\bra{l,m},
\end{equation}
for Clebsch-Gordan coefficient
\begin{equation}
\label{eq:clebschgordan}
       C^{\pm}_{lm}
    :=\sqrt{l(l+1)-m(m\pm1)}.
\end{equation}
The case $l=0$ yields the trivial representation for~$\Omega^{(1)}$
in Eq.~(\ref{eq:OmegaSO(3)}).
\paragraph{Qudit SO(3)-covariant generalized-extreme channels:}
For qudit channels, the dimension of both~$D^{(1)}$ and~$D^{(2)}$ is~$d$.
According to Remark~\ref{remark: UniEquivReps},
amongst uncountable $d$-dimensional unitary representations of SO(3), we only consider one of  many inequivalent representations of SO(3).
Only a finite number of inequivalent representations is possible as the number of inequivalent representations
equals the number of partitions of~$d$
into odd-number components,
such as~$1,1,1$ and~$3$ for $d=3$.
According to Lemma~\ref{lemma:NecCondOmega}, all $SO(3)$ irreps with dimension less than or equal to $d$
label a generalized-extreme channel.
Hence, for $\Omega$,
we have $[\frac{d+1}{2}]$ candidates.
\paragraph{Special case:}
Among the many candidates for $D^{(1)}$,~$D^{(2)}$ and $\Omega$,
we choose one candidate as an example of an SO(3)-covariant generalized-extreme channel.
Our arbitrary choice for study is $\Phi_{\left[\Omega^{(d)}\right],\left[\Omega^{(d)}\right],\left[\Omega^{(d)}\right]}$

In this example~$D^{(1)}$ and~$D^{(2)}$ are chosen to be $\Omega^{(d)}$,
which is an SO(3) irrep.
As SO(3) irreps are odd-dimensional, bearing in mind that the dimension of~$D^{(1)}$ and~$D^{(2)}$ equals the Hilbert-space dimension~$d$, this example is valid just for odd~$d$. 
For the case of $\Phi_{\left[\Omega^{(d)}\right],\left[\Omega^{(d)}\right],\left[\Omega^{(d)}\right]}$,
following Remark~\ref{remark:commutator}, Eq.~(\ref{eq:CovariantLieKraus}) simplifies to the commutators~(\ref{eq:CovariantLieKrauscommutator})
\begin{equation}
\label{eq:commutator_so3}
    \left[\Omega^{(d)}(L_z),A_m\right]
        =mA_m,\;
   \left[\Omega^{(d)}(L_{\pm}),A_m\right]
        =C^{\pm}_{lm}A_{ m\pm1},\;
    l\equiv\frac{d-1}2,
\end{equation}
which yields Kraus operators of $\Phi_{\left[\Omega^{(d)}\right],\left[\Omega^{\left(d\right)}\right],\left[\Omega^{\left(d\right)}\right]}$
being rank-$l$ irreducible spherical tensors
as candidates for generalized extreme channels.
Furthermore, as in this example~$D^{(1)}$ is chosen to be an irrep of $SO(3)$,~$\Xi$ as defined in Eq.~(\ref{def:Xi}) is proportional to the identity~\cite{KMM11}.
To see if the trace-preserving condition~(\ref{eq:TP}) is satisfied,
and to determine if $\Phi_{\left[\Omega^{(d)}\right],\left[\Omega^{(d)}\right],\left[\Omega^{(d)}\right]}$ is extreme or quasiextreme,
we restrict our attention to the low-dimensional cases of $d=3$ and $d=5$ 
for which the problem is tractable. 

%To check the extremality of the channel we should show that if 
%\begin{equation}
%\label{eq:arecursion}
%    \sum_{m,m'=-l}^l\gamma_{m,m'}
%        A_{lm}^{\dag}A_{l,m'}
%          =0,
%\end{equation}
%corresponding to $(2l+1)^2$ linearly coupled equations
%with a single  parameter~$a^{(l)}_{-l}\in\mathbb{R}$
%that must be chosen judiciously
%so that the trace-preserving condition holds;
%then $\gamma_{m,m'}=0$ for all $m$ and $m'$
%\begin{align}
%\label{eq:sum=0}
%    \sum_{m,m'}&\gamma_{m,m'}A_{lm}^{\dag}A_{l,m'}
%                   \nonumber\\
%        =&\sum_{m,m'}     \sum_{n,n'}\gamma_{m,m'}a_{n'}^{(m')}
%            a_n^{(m)*}
%                \ket{l,n}
%                    \nonumber\\
%    &\times\Braket{l,n+m|l,n'+m'}\bra{ l,n'}
%        \nonumber\\
%    =&\sum_{m,m'}\sum_n\gamma_{m,m'}a_{n+m-m'}^{(m')}
%    a_n^{(m)*}\ket{l,n}
%    \bra{l,n+m-m'}
%        \nonumber\\
%    =&0
%\end{align}
%\textbf{Mathematical problem:}
%Given $l\in\mathbb{Z}^+$,
%and~$\{a_n^{(m)}\}$
%given by Eqs.~(\ref{eq:arecursionl}) and~(\ref{eq:arecursion}),
%and equality to zero
%in Eq.~(\ref{eq:sum=0}),
%determine whether a nontrivial set~$\{\gamma_{m,m'}\}$
%(at least one member of the set is nonzero)
%exists.\\
%\textbf{Remark:}
%If a nontrivial solution exists, then the %channel is not extreme.
\paragraph{Qutrit SO(3)-covariant extreme channels:}
\label{par:SO3qutri}
We consider the special case of odd-dimensional Hilbert space for $d=3$.
The number~3 is partitioned into $(3)$,
$(2,1)$ and $(1,1,1)$,
which can be represented by Young diagrams~\cite{Young}
$\Yvcentermath1\Yboxdim{4pt}\yng(3)$,
$\Yvcentermath1\Yboxdim{4pt}\yng(2,1)$,
and
$\Yvcentermath1\Yboxdim{4pt}\yng(1,1,1)$,
respectively.
For our case, only partitions of $d=3$ into odd-number components~$\Yvcentermath1\Yboxdim{4pt}\yng(3)$
and~$\Yvcentermath1\Yboxdim{4pt}\yng(1,1,1)$ are required.
Thus, only two
candidates for each~$D^{(1)}$ and $D^{(2)}$ are possible,
namely three-dimensional representations
\begin{equation}
\Omega^{\Yvcentermath1\Yboxdim{4pt}\yng(1)}\oplus\Omega^{\Yvcentermath1\Yboxdim{4pt}\yng(1)}\oplus\Omega^{\Yvcentermath1\Yboxdim{4pt}\yng(1)},\;\Omega^{\Yvcentermath1\Yboxdim{4pt}\yng(3)}.
\end{equation}
For~$\Omega$ two candidates arise,
namely, a one-dimensional irrep necessarily labeling a unitary channel and a three-dimensional irrep.
Hence, the case we study here,
\begin{equation}
    \Phi_{\left[\Omega^{\Yvcentermath1\Yboxdim{4pt}\yng(3)}\right],\left[\Omega^{\Yvcentermath1\Yboxdim{4pt}\yng(3)}\right],\left[\Omega^{\Yvcentermath1\Yboxdim{4pt}\yng(3)}\right]},
\end{equation}
is one among eight possibilities. According to Eq.~(\ref{eq:commutator_so3}),
Kraus operators are irreducible spherical tensors of rank one, namely,
\begin{equation}
\label{eq:qutritSO3channel}
    A_1=\begin{pmatrix}
    0&-a&0\\
    0&0&-a\\
    0&0&0
\end{pmatrix}
    =-A_{-1}^\dagger,\;
    A_0=\operatorname{diag}\left(a,0,-a\right),\; a\in\mathbb{C}.
\end{equation}
It is straightforward to show that~$\Xi$,
defined in Eq.~(\ref{def:Xi}),
is given by
\begin{equation}
\label{eq:qutritSO3channelsum}
    \Xi=\sum_{k=-1}^1A_k^\dagger A_k
        =2|a|^2\mathds1,
\end{equation}
which satisfies the trace-preserving condition~(\ref{eq:TP}) for $|a|=\frac{1}{\sqrt{2}}$. Furthermore, the set of operators $\mathcal{S}$~(\ref{eq:SAA}) in this example is a set of linearly independent operators; hence, 
according to Theorem~\ref{theorem:Choi},
the qutrit channel $\Phi_{\left[\Omega^{\Yvcentermath1\Yboxdim{4pt}\yng(3)}\right],\left[\Omega^{\Yvcentermath1\Yboxdim{4pt}\yng(3)}\right],\left[\Omega^{\Yvcentermath1\Yboxdim{4pt}\yng(3)}\right]}$ described by Kraus operators~(\ref{eq:qutritSO3channel}) with $|a|=\frac{1}{\sqrt{2}}$ is extreme. 
%---------------
\paragraph{Ququint SO(3)-covariant extreme} channels:
For ququint channels,
the three candidate for each~$D^{(1)}$ and~$D^{(2)}$ are
\begin{equation}
\Omega^{\Yvcentermath1\Yboxdim{4pt}\yng(1)}\oplus\Omega^{\Yvcentermath1\Yboxdim{4pt}\yng(1)}\oplus\Omega^{\Yvcentermath1\Yboxdim{4pt}\yng(1)}\oplus\Omega^{\Yvcentermath1\Yboxdim{4pt}\yng(1)}\oplus\Omega^{\Yvcentermath1\Yboxdim{4pt}\yng(1)},\;
\Omega^{\Yvcentermath1\Yboxdim{4pt}\yng(1)}\oplus\Omega^{\Yvcentermath1\Yboxdim{4pt}\yng(1)}\oplus\Omega^{\Yvcentermath1\Yboxdim{4pt}\yng(3)},\;\Omega^{\Yvcentermath1\Yboxdim{4pt}\yng(5)}.
\end{equation}
For~$\Omega$ there are three candidates, one-, three- and five-dimensional irreps of $\mathfrak{so}(3)$,
namely, $\Omega^{\Yvcentermath1\Yboxdim{4pt}\yng(1)}$, $\Omega^{\Yvcentermath1\Yboxdim{4pt}\yng(3)}$  and $\Omega^{\Yvcentermath1\Yboxdim{4pt}\yng(5)}$. 
Amongst all 27 cases, we focus on the one case 
$\Phi_{\left[\Omega^{\Yvcentermath1\Yboxdim{4pt}\yng(5)}\right],\left[\Omega^{\Yvcentermath1\Yboxdim{4pt}\yng(5)}\right],\left[\Omega^{\Yvcentermath1\Yboxdim{4pt}\yng(4)}\right]}$.
Following the same steps as for qutrit SO(3)-covariant channels,
but here for $d=5$, then $\Phi_{\left[\Omega^{\Yvcentermath1\Yboxdim{4pt}\yng(5)}\right],\left[\Omega^{\Yvcentermath1\Yboxdim{4pt}\yng(5)}\right],\left[\Omega^{\Yvcentermath1\Yboxdim{4pt}\yng(4)}\right]}$ is described by
five Kraus operators that are irreducible rank-two spherical tensors
\begin{equation}
\label{eq:A2A1}
    A_2=\frac12\begin{pmatrix}
0&0&2a&0&0\\
0&0&0&\sqrt{6}a&0\\
0&0&0&0&2a\\
0&0&0&0&0\\
0&0&0&0&0
\end{pmatrix},\;
A_1=\frac12\begin{pmatrix}
0&-\sqrt{6}a&0&0&0\\
0&0&-a&0&0\\
0&0&0&a&0\\
0&0&0&0&\sqrt{6}a\\
0&0&0&0&0
\end{pmatrix}
\end{equation}
and
\begin{equation}
\label{eq:A0A-1A-2}
    A_0=\frac{a}{2}
        \operatorname{diag}\left(2,-1,-2,-1,2\right),\;
    A_{-1}=-A_1^\text{T},\;
    A_{-2}=A_2^\text{T},\;a\in\mathbb{C},
\end{equation}
which satisfy the trace-preserving condition~(\ref{eq:TP}) for  $|a|=\sqrt{2/7}$.  
The operators in $\mathcal{S}$~(\ref{eq:SAA}) are linearly independent;
hence, the ququint channel $\Phi_{\left[\Omega^{\Yvcentermath1\Yboxdim{4pt}\yng(5)}\right],\left[\Omega^{\Yvcentermath1\Yboxdim{4pt}\yng(5)}\right],\left[\Omega^{\Yvcentermath1\Yboxdim{4pt}\yng(4)}\right]}$ is extreme.  
%--------------------------------------------------------------
\subsubsection{$SU(2)$-covariant channels}
\label{subsubsec:SU(2)}
%\paragraph{About:}
In this subsubsection we consider SU(2)-covariant channels.
As~$\mathfrak{su}(2)\simeq\mathfrak{so}(3)$
discussed in~\S\ref{subsubsec:SO3},
much of the analysis here is similar,
but SU(2) admits both even- and odd-dimensional irreps
whereas SO(3) only admits odd-dimensional irreps.
The ladder and Cartan operators are~$J_\pm$
and~$J_z$, respectively,
and the Casimir invariant is~$J^2$
with spectrum~$\{j(j+1);2j\in\mathbb{N}\}$. Unitary irreps of $\mathfrak{su}(2)$
are given by 
\begin{equation}
\label{eq:OmegaSU(2)}
    \Omega^{(2j+1)}(J_\pm)=\sum_{m=-j}^{j}C^{\pm}_{jm}\ket{j,m\pm1}\bra{j,m},\;
    \Omega^{(2j+1)}(J_z)=\sum_{m=-j}^{j} m\ket{j,m}\bra{j,m}
\end{equation}
with~$\{\ket{j,m}\}$ the weight states labelled by Casimir and Cartan operators, respectively,
and
\begin{equation}
       C^{\pm}_{jm}
    =\sqrt{j(j+1)-m(m\pm1)}
\end{equation}
the Clebsch-Gordan coefficients. 
As in~\S\ref{subsubsec:SO3},
the case $j=0$ yields the trivial representation for~$\Omega^{(1)}$
in Eq.~(\ref{eq:OmegaSU(2)}).
\paragraph{Qudit SU(2)-covariant generalized-extreme channels:}
For qudit channels, the dimension of both~$D^{(1)}$ and~$D^{(2)}$ is~$d$.
According to Remark~\ref{remark: UniEquivReps},
amongst uncountable $d$-dimensional unitary representations of SU(2), we only consider one of  many inequivalent representations of SU(2). Similar to the case of SO(3)-covariant qudit channels, the number of inequivalent $d$-dimensional representation of SU(2) is finite. For SO(3) the total number equals the number of partitions of~$d$ into odd-number components, but,
for SU(2), it equals to the number of partitions of~$d$ into even- and odd-number components because irreps of~$\mathfrak{su}(2)$ have both even- and odd-dimensional irreps.  
As an instance $\Yvcentermath1\Yboxdim{4pt}\yng(3)$,
$\Yvcentermath1\Yboxdim{4pt}\yng(2,1)$,
and
$\Yvcentermath1\Yboxdim{4pt}\yng(1,1,1)$ are partitions for $d=3$.
According to Lemma~\ref{lemma:NecCondOmega}, all SU(2)-covariant generalized-extreme channels are labelled by irreps of $\mathfrak{su}(2)$ with dimension less than or equal to~$d$.
Hence, for $\Omega$,
we have~$d$ candidates.
\paragraph{Two cases:}
For $d$-dimensional Hilbert space~$\mathscr{H}_d$ amongst all candidates, we focus on two cases to illustrate salient points.
In the first case we clarify how the special case studied for SO(3)-covariant channels are also SU(2)-covariant channels.
In the second case we present SU(2)-covariance,
which we prove to be extreme for any dimension~$d$. 
\begin{kase}
\label{case:su2first}
Amongst all possible candidate,
we pick
$\Phi_{\left[\Omega^{(d)}\right],\left[\Omega^{(d)}\right],\left[\Omega^{(d)}\right]}$
For this case,
following Remark~\ref{remark:commutator}, Eq.~(\ref{eq:CovariantLieKraus}) simplifies to the commutators~(\ref{eq:CovariantLieKrauscommutator})
\begin{equation}
\label{eq:commutator_su2}
   [\Omega^{(d)}(J_z),A_m]
        =mA_m,\;
   [\Omega^{(d)}(J_{\pm}),A_m]
        =C^{\pm}_{jm}A_{ m\pm1}.
\end{equation}
For odd values of $d$, that is for $j\in\mathbb{N}$, Eq.~(\ref{eq:commutator_so3}) yields Kraus operators of $\Phi_{\left[\Omega^{(d)}\right],\left[\Omega^{\left(d\right)}\right],\left[\Omega^{\left(d\right)}\right]}$
being rank-$j$ irreducible spherical tensors, which is exactly the same as what we have for SO(3)-covariant channels $\Phi_{\left[\Omega^{(d)}\right],\left[\Omega^{\left(d\right)}\right],\left[\Omega^{\left(d\right)}\right]}$. 
When~$d$ is even, that is $2j\in\mathbb{N}$, Eq.~(\ref{eq:commutator_su2}) yields $A_{-j\leq m\leq j}=0$. 
Thus, if,
among all candidates,
we restrict our attention to $\Phi_{\left[\Omega^{(d)}\right],\left[\Omega^{\left(d\right)}\right],\left[\Omega^{\left(d\right)}\right]}$ we get nothing more than the case we studied for SO(3)-covariant channels on odd-dimensional Hilbert space. $\blacksquare$
\end{kase}
\begin{kase}
\label{ex:su2second}
Amongst the candidates,
we pick
$\Phi_{\left[\Omega^{(d-1)}\oplus \Omega^{(1)}\right],\left[\Omega^{(d-1)}\oplus \Omega^{(1)}\right],\left[\Omega^{(d-1)}\right]}$ for
$\Omega^{(d-1)}$ defined in~(\ref{eq:OmegaSU(2)}) with $j=\frac{d-2}2$.
Following Remark~\ref{remark:commutator}, Eq.~(\ref{eq:CovariantLieKraus}) simplifies to the commutators~(\ref{eq:CovariantLieKrauscommutator})
\begin{equation}
\label{eq:commutator_su2D}
   \left[\Omega^{(d-1)}(J_z)\oplus \Omega^{(1)}(J_z),A_m\right]
        =mA_m,\;
    \left[\Omega^{(d-1)}(J_\pm)\oplus \Omega^{(1)}(J_\pm),A_m\right]
        =C^{\pm}_{jm}A_{ m\pm1}.
\end{equation}
The Kraus-operator solution to Eq.~(\ref{eq:commutator_su2D}) is~\cite{KM08}
\begin{equation}
\label{eq:Asm}
    A_m=\frac{1}{\sqrt{d-2}}\ket{j=(d-2)/2,m}\bra{e}
    +(-1)^{\frac{d-2}2-m}\ket{e}\langle j=(d-2)/2,-m|,
\end{equation}
represented in the orthonormal basis of Hilbert space
\begin{equation}
\mathscr{H}=\mathscr{H}_{d-1}\oplus\mathscr{H}_1,\;
\mathscr{H}_{d-1}
    =\operatorname{span}
    \{\ket{j=(d-1)/2,m}\},\;
\mathscr{H}_1
    =\operatorname{span}\{\ket{e}\}.
\end{equation}
The trace-preserving condition (\ref{eq:TP}) is evidently satisfied:
\begin{equation}
\Xi=\sum_{m=-j}^j   
    A_m^{\dag}A_m
        =\mathds1_{d-1}\oplus\ket{e}
            \bra{e}
                =\mathds1_d.
\end{equation}
To show that $\Phi_{\left[\Omega^{(d-1)}\oplus \Omega^{(1)}\right],\left[\Omega^{(d-1)}\oplus \Omega^{(1)}\right],\left[\Omega^{(d-1)}\right]}$ is extreme, we prove that the operators in $\mathcal{S}$~(\ref{eq:SAA})
are linearly independent.
Hence, we assume that 
\begin{equation}
\sum_{m,n=-j}^j\alpha_{mn}A_m^{\dag}A_{n}=0,\;
\alpha_{mn}\in\mathbb{C},
\end{equation}
and our task is to prove that $\alpha_{mn}\equiv0$
for all~$m$ and~$n$.
By replacing $A^{\dagger}_m$ and $A_n$ from Eq.~(\ref{eq:Asm}),
we have
\begin{equation}
    \sum_{m,n=-j}^j\alpha_{mn}A_m^{\dag}A_{n}
    =\sum_{m=-j}^j\frac{\alpha_{mm}}{d-1}\ket{e}\bra{e}
    +\sum_{m,n=-j}^j\alpha_{mn}(-1)^{m+n}\ket{j,-m}\bra{j,-n}=0.
\end{equation}
Therefore, $\alpha_{mn}=0$ for all $m$ and $n$,
which implies that 
the channel $\Phi_{\left[\Omega^{(d-1)}\oplus \Omega^{(1)}\right],\left[\Omega^{(d-1)}\oplus \Omega^{(1)}\right],\left[\Omega^{(d-1)}\right]}$ described by Kraus operators $A_m$~(\ref{eq:Asm}) is an extreme channel
for all~$d$.
We now illustrate explicitly for three subexamples.

%---------------------------
\paragraph{Qubit channel:}
Now we consider $d=2$,
hence, $j=\tfrac12$,
for qubits and we continue the reasoning above. Channel $\Phi_{\left[\Omega^{(1)}\oplus \Omega^{(1)}\right],\left[\Omega^{(1)}\oplus \Omega^{(1)}\right],\left[\Omega^{(1)}\right]}$, is labelled by
one-dimensional irrep of SU(2).
Hence, it has just one Kraus operator and is a unitary channel. Following Eq.~(\ref{eq:Asm}) its single Kraus operator is given by
\begin{equation}
    A_0=\ket{0,0}\bra{e}+\ket{e}\bra{0,0},
\end{equation}
which is one of the extreme channels discussed in $Z_2$-covariant channels in \S\ref{subsec:Z2} and is already discussed in the literature~\cite{R00, Cho75}.
%--------------------
\paragraph{Qutrit channel:}
Continuing with cases of the second example,
this time for qutrits,
we let $d=3$
and thus set
$j=\nicefrac12$.
According to Eq.~(\ref{eq:Asm}), the extreme channel $\Phi_{\left[\Omega^{(2)}\oplus \Omega^{(1)}\right],\left[\Omega^{(2)}\oplus \Omega^{(1)}\right],\left[\Omega^{(2)}\right]}$ 
is described by a pair of Kraus operators
\begin{equation}
    A_{\pm \nicefrac{1}{2}}
        =\frac{1}{\sqrt2}
        \ket{\nicefrac12,\pm\nicefrac12}\bra{e}\pm\ket{e}\bra{\nicefrac12,\mp\nicefrac12}.
\end{equation}
%---------------------
\paragraph{Ququart channel:}
Continuing with the second example,
for $d=4$ we set $j=1$
and the number of Kraus operators $K=3$.
According to Eq.~(\ref{eq:Asm}),
Kraus operators of the extreme ququart channel $\Phi_{\left[\Omega^{(3)}\oplus \Omega^{(1)}\right],\left[\Omega^{(3)}\oplus \Omega^{(1)}\right],\left[\Omega^{(3)}\right]}$ are
\begin{equation}
A_{\pm 1}=\frac13\ket{1,\pm1}\bra{e}+\ket{e}\bra{1,\mp1},\;
A_0=\frac13\ket{1,0}\bra{e}-\ket{e}\bra{1,0}.
\end{equation}
These illustrations show how our formulation applies for all~$d$.
$\blacksquare$
\end{kase}
\section{Discussion}
\label{sec:discussion}
Now we summarize and discuss our results.
We begin with discussing group-covariant generalized-extreme channels.
Then we explain the nature and importance of our pseudocode method and results.
Finally, we discuss our explicit examples and their generalizations.

Our first results concerned establishing the mathematical framework for constructing group-covarariant extreme and quasi-extreme channels.
Our approaches leverages off the well understood case of qubit channels,
but only a few extreme cases are known for dimension $d>2$
without providing insight into how to extend beyond these examples.
The full problem of characterizing and constructing extreme channels is too daunting
so we restrict our attention to a subset of extreme channels that are group-covariant
for the group either being a finite discrete group or else a compact connected Lie group.
By studying this subset,
we make some of the hard problems concerning constructing extreme and quasi-extreme channels tractable and, furthermore,
group-covariant channels are useful and valuable in their own right.

Although previous results concern specific examples of generalized-extreme channels,
we introduce a systematic method for constructing all generalized-extreme channels if they are covariant with respect to finite discrete or compact connected Lie groups.
Our method labels every group-covariant channel with three unitary representations of the channel:
the channel is group-covariant with respect to the first two labels and the third label removes multiplicity.
Multiplicity is removed by uniquely labelling each channel adding a label whose purpose is to distinguish between different channels that are identically labelled with respect to the first two unitary representations. 
After uniquely labelling channels,
we prove that any group-covariant channel is a generalized-extreme channel
if and only if its third label is a group irrep whose dimension does not exceed the Hilbert-space dimension.
With these results,
we have set the stage for a systematic method,
which we formalize as algorithms expressed in pseudocode.

As algorithms and pseudocode are not common in studies of quantum information theory,
we carefully developed the relevant concepts and explained how our pseudocode works in formally presenting algorithms.
Importantly,
our algorithms are built,
not a Turing-type computational model but rather on a computational model involving complex-number arithmetic at the foundational level.
Our choice of computational model reflects that our results are aimed not at solving generalized-extreme channels using computers but rather formalizing the problems and their solutions as logical steps for mathematical physics.

Thus, our systematic method for solving generalized-extreme channels is about constructing solutions for every conceivable group-covariant channel with groups of particular classes, namely, finite discrete groups and compact connected Lie groups. Specifically, we establish a procedure to construct group-covariant generalized-extreme channels methodically given the name of the group and the Hilbert-space dimension.
To this end,
we formalize existing knowledge about groups in terms of a (hypothetical but plausible) library,
which is the repository of information about groups and is called upon in our pseudocode.
Our systematic method guarantees,
for any Hilbert-space dimension and any valid group name,
an output comprising the set of all group-covariant generalized-extreme channels,
which,
in the empty-set case,
implies non-existence of any group-covariant generalized-extreme channels for that group at that dimension.
Furthermore, we present pseudocode for deciding whether a given channel is extreme or quasi-extreme. Therefore, starting with a given dimension and a group name, 
and employing our first and second algorithms expressed in pseudocode,
we show how to methodically construct the entire set of generalized extreme channels with an additional label conveying whether this generalized-extreme channel is extreme or quasi-extreme.

We then present examples showing the application of our systematic method to these instances.
These examples illustrate our methods to elucidate how our techniques work and furthermore validate our approach by showing that known results,
specifically for two-dimensional $\mathbb{Z}_2$-covariant channels,
are obtained using our technique.
We also obtain novel results for new group-covariant examples, which show interesting results such as, for three-dimensional $S_3$-covariant channels,
we obtain a continuously parametrized family of extreme channels instead of a finite number of extreme channels.
Another intriguing result we find is that the extreme channel output is never empty for the group SU(2) and for any dimension greater than equal to two.

Furthermore, we observe that,
as the group becomes larger, the number of candidates for non-unitary generalized extreme channels tends to increase,
but, importantly, 
exceptions exist to this rule of thumb such as that shown in our results for the~$D_5$ case.
This general rule and its exceptions highlights the importance of using pseudocode for constructing group-covariant generalized-extreme channels,
as studying this problem analytically is infeasible.
Another important point raised in the last example concerning SU(2)-covariance
shows that
our algorithm yields a non-empty set of extreme channels
for all Hilbert-space dimensions without exception.
\section{Conclusions}
\label{sec:conclusions}
%\paragraph{The problem:}
We have addressed the problem of constructing the set of extreme channels for $d$-dimensional Hilbert-space. 
This well known problem is hard because the set of channels at Hilbert-space dimension $d>2$
has not been parameterized.
Therefore, the detailed structure of the set of channnels
and a parameterized description of the boundary of the set are unknown,
which makes it impossible currently to construct directly the extreme points of the set of channels.
By considering symmetry,
we can construct a subset of extreme channels,
and convex combinations of these extreme channels enable parametrization of an important subset of channels. 

%\paragraph{What we did:}
Here we have restricted our attention to a subset of extreme channels that are covariant with respect to a finite discrete group or a compact connected Lie group.
By exploiting knowledge about group and representation theory,
we are able to develop a systematic approach to construct those extreme and quasi-extreme channels that are group-covariant. 
Our systematic approach
is represented by pseudocode, 
which makes our procedure especially clear.
We present a variety of elucidating and instructive examples,
which includes the proof that an extreme group-covariant channel exists for every choice of Hilbert-space dimension.

%\paragraph{Importance:}
Our results extend significantly knowledge of extreme channels by going beyond the usual restriction to unital channels for $d>2$.
Furthermore,
our approach reveals that the problem of constructing group-covariant generalized-extreme channels
reduces to the well-studied problems of solving a system of linear and quadratic equations. 

%\paragraph{Outlook:}
At this stage,
our theory does not yet consider tensor products of Hilbert spaces.
Our results could be generalized in nontrivial ways by dealing with tensor products,
perhaps restricting to the same groups considered here to study group-covariant extreme channels involving this added structure.
Incorporating tensor-product structure could be useful for solving problems of correlated channels~\cite{DuanGuo1998, MacchiavelloPalma2002}.

One interesting application of our algorithmic approach could be to the problem of searching for increasingly large Holevo-capacity additivity violations;
regarding the conjecture that Holevo capacity of quantum channels is additive~\cite{Hol06},
Hastings disproves this conjecture  by introducing a counterexample but leaves open how to find all channels that are non-additive and how large of an additivity violation is possible~\cite{Has09}.
Here we suggest an algorithmic approach,
based on our methods,
to addressing this problem of increasing the violation.

Now we introduce additivity violation as a superadditivity computational problem of discovering channels with greater additivity violation compared to what is known now.
%Our results can be used to
%identify pair of group-covariant generalized extreme channels that violate the additivity of Holveo capacity in the strongest possible way.
Additivity violation of Holevo capacity for a pair of quantum channels~$\Phi_{1,2}$ is quantified by
\begin{equation}
\label{eq:v}
    v(\Phi_i,\Phi_2):=\chi(\Phi_1\otimes\Phi_2)-\chi(\Phi_1)-\chi(\Phi_2),\,
    \chi(\Phi)
    =\sup_{\{p_i,\rho_i\}}\left(\mathscr{S}\left(\sum_i p_i\Phi(\rho_i))-\sum_i p_i \mathscr{S}(\Phi(\rho_i))\right)\right)
\end{equation}
if~$v$ is positive,
with $\rho_i\in\mathcal{T}(\mathscr{H})$, $p_i>0,\;\forall i$, $\sum_i p_i=1$ and $\mathscr{S}(\rho)$ the von~Neumann entropy of $\rho$.
An algorithm for computing $v$~(\ref{eq:v})
would accept the descriptions of two channels~$\Phi_{1,2}$
and yield~$v$ as output:
a value of~$v$ greater than the best-known~$v$ to date would be flagged as a success.
Our algorithm yields descriptions of generalized extreme channels as outputs so the algorithm for computing~$v$ would call our algorithm as an oracle to obtain pairs of channels,
which could be for different Hilbert-space dimension~$d$ and different group names~$\{\mathcal{G}\}$.

Using this approach, we compute the additivity violation $v$~(\ref{eq:v}) for pairs of channels drawn from the set of group-covariant generalized extreme channels. 
This set is smaller than the set of all channel pairs that should be searched for obtaining the largest possible~$v$,
but this restricted search is a good start to search algorithmically for the largest~$v$ over all possible channel pairs.

We emphasize some caveats on our approach to discovering increasingly large~$v$ algorithmically.
As our algorithm is designed for the Blum-Shub-Smale machine~\cite{BSS},
adapting to a Turing type of discrete computer is needed:
this adaptation is achieved by working with floating numbers with a consequence that computations are then approximate rather than exact.
Second,
our algorithm does not generate all possible pairs of channels but rather just a restricted set of channels, so our additivity-violation algorithm would not be performed exhaustively over all pairs of channels but,
if an algorithm could be devised to generate all possible channel pairs, that algorithm would supplant our own generalized extreme channel generator and then permit an exhaustive search.
We also emphasize that our algorithm could be infeasible on current computers.

Another promising direction to follow would be to extend beyond pseudocodes to writing actual computer code and implementing on a computer to solve for new group-covariant extreme channels.
Our results could be used to construct circuits for simulating extreme channels,
which could be useful for quantum-channel simulation theoretically~\cite{WBOS13,ItenColbeckChristandl2017, BGNPZ13} and experimentally~\cite{Jeong2013, Shaham12, Sciarrino2004, Lu2017}.
Finally,
our theory could help to study complexity considerations associated with quantum-circuit simulation of channel-construction problems~\cite{WBOS13,BGNPZ13,ItenColbeckChristandl2017,WS15}.
\acknowledgments
L.\ M.\ acknowledges financial support by Sharif University of Technology, Office of Vice President for research under Grant No.\ G930209 and hospitality by the University of Calgary where parts of this work were completed.
B.\ C.\ S.\ appreciates financial support from NSERC.
Both L.\ M.\ and B.\ C.\ S.\ acknowledge financial support from the American Physical Society through their International Research Travel Award Program,
which supported a research visit to the University of Calgary.
We acknowledge valuable discussions with Camila Suarez Viltres concerning pseudocode and data-type formalities.
\appendix

\section{Formal approach}
\label{subsec:formalapproach}
In this appendix we review the formal approach and main ingredients for specifying a formal problem, 
as well as basics for writing a formal problem as a pseudocode. We first discuss,
in~\S\ref{subsubsec:formalproblems},
how we specify problems
formally as such formal specifications are needed for our algorithmic approach,
and we make clear our model for solving formal problems.
Second we explain pseudocode, data types, and the notation we use. Next we provide background on the role of a library in writing algorithms to solve the problem. Finally, we explain the structure of any algorithm and its essentials.
\subsection{Formal problems}
\label{subsubsec:formalproblems}
In this subsection,
we provide principles for specifying formal problems.
Then we explain our model for stating and solving formal problems,
with our model described in computational terms.
%\paragraph{Stating formal problem:}
A problem should be stated in a form that the reader can understand,
which requires that the language for the problem be clear and the terms employed are understood by the intended audience,
which can be sentient (e.g., a human) or autonomous (e.g., a computer).
In our case we would like our hard channel problems to be so clear that a machine can understand and potentially solve the problem.
For this goal of autonomous solving to be met,
we adopt principles from computing,
which requires us to specify various structures such as
data types, which could be integer or real or complex numbers or symbols
as we explain below.
A formal problem has well-defined arithmetic operations such as addition and multiplication.
The formal statement is lucid regarding input and output and whether the solution is obtained deterministically or probabilistically and whether exact or approximate.
The procedure to map input to output is specified in terms of universal primitive operations.
Library functions are permitted;
such library functions offer known algorithms that solve specific problems and are in a library because these algorithms are popular for multiple applications.

%\paragraph{Model:}
Typically the Turing machine~\cite{Turing1937},
or perhaps its extension to quantum Turing machines~\cite{Yao1993, MolinaWatrous2018},
would suffice as the starting point for the foundational model on which to base formal problems and their algorithmic solutions.
However, we are not devising ways to solve group-covariant channels on a computer per se but rather seeking to formalize the mathematical problems.
As these quantum channels are defined over the complex number field,
we choose as our starting point the Blum-Shub-Smale machine~\cite{BSS},
which modifies the Turing machine to allow for an uncountable alphabet corresponding to real numbers.
Although we are employing complex numbers,
rather than real numbers,
the Blum-Shub-Smale machine can be extended to complex numbers by treating complex numbers as pairs of real numbers with complex conjugation and multiplication rules suitably incorporated.
In addition to the alphabet including the full uncountable set of complex numbers,
our alphabet also includes standard alphanumeric characters.

%\paragraph{Specifying and solving problems}
%\label{para:computationalproblems}

Various types of problems can be specified
such as search, optimization, function inversion and decision.
In our work,
we deal with just two types of computational problems,
namely function problems and decision problems.
A function problem maps each input to some output according to a set of rules,
such as exponentiating or taking a square root.
Of course a well posed function problem can fail to have a solution that occurs from dividing by zero.
A decision problem is a special case of a function problem in that the output can 
only be binary, such as `true' denoted by $\top$ or `false' which is denoted by $\bot$.

A problem statement requires specification of a clear task and needs well defined inputs and outputs including their data types.
The algorithm for solving the problem is expressed as inputs, outputs, and a step-by-step procedure to map the input to the output with these steps specified in terms of operations that are themselves constructed from a universal set of primitives,
and some operations can be drawn from the library whereas others are explicit logical steps in terms of the computational model.

%***********
\subsection{Essentials of algorithms}
\label{subsubsec:essentialsalgorithms}
%\paragraph{About:}
An algorithm is designed to solve a specific, well posed problem, in terms of instructions that the machine can understand.
Specifically,
the algorithm accepts inputs,
processes these inputs by a procedure expressed as a sequence of instructions, and yields the desired output.
Below we describe what inputs, outputs and procedures are as we use this terminology and employ these principles throughout our work.
\paragraph{Inputs:}
The first component of an algorithm is its inputs.
We require, in our approach to algorithms,
the value of the input and also type.
Type refers to the nature of this input, which implies rules such as arithmetic or concatenation.
In some programming languages,
such as Python to which we refer often,
type is not needed but then an interpreter is required;
as we are focused on an algorithmic approach but not on actual programming or implementations of computation,
we prefer to keep all types explicitly stated.

\paragraph{Outputs:}
The outputs of the algorithm are the solutions of the problem.
In our approach we specify type of each output to be clear, and the format is the same as for inputs.
Outputs can be exact or approximate solutions and obtained deterministically or probabilistically.
Flags can be binary outputs that indicate whether a valid solution was found or not.
\paragraph{Procedures:}
The algorithmic procedure's purpose is to map the inputs to the outputs in a logical way following steps that the machine can understand.
Procedures can be represented by flow charts or by pseudocode;
we prefer using pseudocode,
which includes both instructions and comments following $\triangleright$ symbols that explain briefly the instructions.
Procedure statements include explicit declaration of variables, including their types,
and return statements for sending variable values to outputs and end statements that terminate the procedure.
We typically do not include initialization statements as we assume all numbers and numerical arrays are initialized to zero~$0$ or arrays of zeros and symbolic variables are initialized to blank~$\flat$.
%**********--------
\subsection{Pseudocode}
\label{subsec:Pseudocode}
%\paragraph{About:}
In this subsubsection, we explain pseudocode. We have thus far explained formal problems and essentials of algorithms,
and now
we explain the transition from algorithms to pseudocode.

Pseudocode appears similar to how a formal program looks,
but pseudocode is not meant to be compiled.
Rather pseudocode formalizes the logic of how we are solving the problem and it is an alternative to a flow-chart representation.
We employ pseudocode to ensure that our algorithms are complete.
Breaking what could be one large algorithm into multiple subalgorithms ensures that what would be a complicated algorithm is modularized and thus fathomable.

Our pseudocode for each algorithm and subalgorithms (treated as an algorithm) is introduced by first presenting the description of the algorithm and then presenting the name of the algorithm.
Then we present sequentially the input and then the output and then the procedure.
The procedure is named on formal line~1 of each pseudocode and is followed by the formal name of the algorithm and then its arguments in parentheses.

In the procedure, the pseudocode begins with importing functions from the library or other algorithms.
Variables that arise during the procedure are declared as they arise.
Requisite variables are returned for output before the procedure ends.

%******************
\section{Notation and datatypes}
\label{subsubsec:datatypes}
%\label{subsubsec:notation}
%\paragraph{About:}
In this appendix, we  establish concepts and notation for those data types~\cite{DataType} that are pertinent to our study and used in our pseudocode.  Our two key additions to the typical notation in physics pseudocode are the use of multi-dimensional arrays with hyperrows and using a jagged structure. We describe these new notational results at the end of this appendix.

We write data types in lower-case \texttt{typewriter} font
and allow for both static and dynamic arrays.

Names of variables are expressed in small-capital \textsc{camelCase}
with our pseudocode presentation and notational convention closely following the qudit benchmarking case~\cite{JWSS20}.
As our model includes real and complex numbers, we use these data types in that context rather than real and complex data types representing floating point numbers as in languages such as FORTRAN.

Now we proceed to introduce and explain specific data types.
In some cases we introduce new data types for use in our studies;
although some terms and types are new,
these novelties are based on established concepts.

\paragraph{Numbers:}
\label{par:DataTypeNumbers}
First we define numbers,
which can be a countable set such as integers~$\mathbb{Z}$,
a finite subset of integers such as~$[n]$~(\ref{eq:Akij}) for~$n\in\mathbb{Z}^+$
(positive integers
denoted \texttt{posinteger})
and natural numbers~$\mathbb{N}$ (\texttt{natural})
comprising positive integers and zero
or the uncountably infinite numbers such as
real~$\mathbb{R}$ (\texttt{real})
and complex~$\mathbb{C}$ (\texttt{complex}) numbers
or their floating-point representations.
A binary digit (bit) has type \texttt{binary} with two values~$0$ and~$1$.
All numbers are subject to arithmetic operations such as addition/subtraction~$\pm$,
multiplication~$*$
and exponentiation~$\wedge$,
with these operations and their character representations regarded as being defined deep in the computer architecture stack~\cite{Knu77}
in a natural way for each kind of number.
\paragraph{Non-numerical:}
Variables are represented by symbols
(denoted \texttt{symbol}),
which serve akin to variables in symbolic expressions.
These symbolic variables are expressed as an alphabetical string of Latin and Greek letters in our model.
The type \texttt{string} is a concatenation of alphanumeric characters, brackets and punctuation marks,
but not arithmetic operation symbols,
to produce a literal constant such as ``SO(3)'' to denote the special orthogonal group of dimension three.
\paragraph{Array:}
\label{para:array}
An array is a list of elements drawn from the same data type
(e.g.,
\texttt{natural}, \texttt{symbol})
and can be multidimensional.
We introduce the term ``hyperrow'' here to refer to coordinates (labelled by indices) of the array,
with a hyperrow of first order referring to an array's row, a
hyperrow of second order referring to a column, a hyperrow of third order referring to, say, a sheet (following Microsoft Excel terminology) and so on. A one-dimensional array is a vector
whose coordinates are row numbers.
A matrix is a two-dimensional array,
with coordinates given by rows and columns.
This concept can be generalized to multi-dimensional arrays with multiple arrays.
An array need not have the same number of columns for each row:
instead the number of columns could vary for each row,
in which case this array is ``jagged''.
In our pseudocode,
we declare 
 arrays following the common convention, that is, expressing first the type of array and its elements data type, second, the dimension of the array represented by the number of brackets, and, finally, the name of the array. For example,
\begin{equation}
\texttt{symjagged}[~][~]~\textsc{arrayOne}    
\end{equation}
declares a two dimensional symbolic jagged array called \textsc{arrayOne}.
As another example,
\begin{equation}
\texttt{complexarray}[2][~]~\textsc{arrayTwo}
\end{equation}
declares a two dimensional complex array called \textsc{arrayTwo}
with its row having two elements and the number of
elements in its column not determined at the declaration step.
We address array elements using the notation
$\textsc{array}[\text{row number}][\text{column number}]$ for an element of the two-dimensional array \textsc{array} with row number and column number. 
This notation is extended in a straightforward way for higher-dimensional arrays.
Generalized arithmetic operations used for arrays include tensor product~$\otimes$
and direct sum~$\oplus$.

Previously,
two-dimensional arrays have been treated in physics pseudocode,
but higher-dimensional arrays,
although employed in programming,
have not previously been introduced in physics pseudocode.
Here we produce a straightforward but useful notation for higher-dimensional arrays,
exemplified by the following case.
A three-dimensional array,
denoted~\textsc{array},
with elements being of any type, say, \texttt{symbol} for symbol is declared as \texttt{symbol}[\textsc{size0}][\textsc{size1}][\textsc{size2}]\textsc{array},
where \textsc{size0}, \textsc{size1} and \textsc{size2}, respectively are the sizes of the first-order hyperrrow, second-order hyperrow, 
and third-order hyperrow,
respectively,
for~\textsc{array}.
By convention, the index of each array starts from zero.
Then $\textsc{array}[\textsc{i}][\textsc{j}][\textsc{k}]$ is an element of \textsc{array} with sheet (third-order hyperrow) number \textsc{k}, column (second-order hyperrow) number \textsc{j} and row (first-order hyperrow) number \textsc{i},
and this notation convention extends in an obvious way to arrays of higher dimensions.

Sometimes we need to work with a lower-dimensional array from a higher-dimensional array,
which we explain here by example.
Given three-dimensional array $\textsc{array}$
defined above,
we construct the `projected' one-dimensional array
$\textsc{oneDArray}\gets\textsc{array}[\textsc{i}][\textsc{j}]$,
which is a one-dimensional array extracted from \textsc{array} with \textsc{size2}.
The~$\textsc{k}^\text{th}$ element of \textsc{oneDArray}, \textsc{oneDArray}[\textsc{k}], is $\textsc{array}[\textsc{i}][\textsc{j}][\textsc{k}]$.
Similarly, $\textsc{twoDArray}\gets\textsc{array}[\textsc{i}]$, is a two-dimensional array extracted from \textsc{array}
with the first hyperrow of size \textsc{size1} and the second hyperrow of size \textsc{size2}. 

We find necessary the use of jagged arrays,
which are used in programming but have not been introduced into pseudocode within physics papers yet
so we do so here.
Hence, we define a multi-dimensional jagged array
as an array whose order~$n$ hyperrow has a size that depends on the index of the hyperrow of order $n-1$.
For example, if \textsc{array} is a three-dimensional jagged array with elements of type symbol,
its declaration is
\texttt{symjagged}[\textsc{size0}][\textsc{size1}][ ]\textsc{array},
and \textsc{array}[\textsc{i}][\textsc{j}]
is a one-dimensional array 
whose size depends on the chosen indices \textsc{i} and \textsc{j}. %*********************************************************

%*********************************************************
\section{Library}
\label{subsubsec:library}
%\paragraph{About:}
In this appendix,
we paraphrase the concept and role of libraries~\cite{Library} in solving computational problems algorithmically.
Then we briefly describe the functions we use for solving group-covariant channels in  \S\ref{subsec:AlgorithmicApproach} and  \S\ref{subsec:pseudocoding}. In \S\ref{subsubsec:blibrary} we explain the common functions available in public libraries that are required in our pseudocode. Then, in \S\ref{subsubsec:rlibrary} we discuss a convenient library of our new functions used in our pseudocode.
%\paragraph{Concept and role of library:}
A library is vital in computer programming and more generally in software development~\cite{Library}.
Libraries appear in a variety of forms,
but our use of libraries here is restricted to subroutines,
which are a packaged sequence of instructions that perform a well defined task~\cite{Knu77}.
Such library elements are known as functions,
and we use this terminology throughout.
A function is called by invoking the name of the function and passing parameters and then receiving new parameter values after execution.
%-----------------
\subsection{Functions available in public library}
\label{subsubsec:blibrary}
In this subsection we describe the common well known functions that are useful for our algorithm. Specifically,
we describe the matrix functions that prove useful later
and 
then we proceed to describe the function we use to solve simultaneously sets of equations.
Finally, we describe the functions that identify symbols in expressions and simplifies algebraic expressions.
Our approach to introducing library functions is to make use of concepts existing in the literature or from actual libraries such as Python's SymPy~\cite{SymPy} or NumPy~\cite{NumPy}
or alternatives such as NAG~\cite{NagLibrary}.

%\paragraph{Array functions:}
As solving linear equations is vital to our analysis,
and array manipulation is germane to solving such systems of equations,
we introduce here basic array library functions that prove to be useful in our algorithm for solving generalized group-covariant channels.
The zero and identity matrices are especially useful, so we introduce library functions \textsc{zero} and \textsc{id}
for creating zero and identity matrices of specific size,
and they can be two-dimensional arrays or beyond to multidimensional cases.
If a matrix has complex entries,
the transpose and the Hermitian conjugate of a matrix can also
be valuable, for which we use the functions \textsc{transp} and \textsc{dag}, respectively.
Another important function for our purposes is \textsc{reshape},
which converts an array of some dimension to an array of another dimension,
with the easiest nontrivial examples being conversion of a vector to a matrix and vice versa.
Although SymPy includes an ordering option,
we only ever use one ordering here, so we do not include this input option.

%\paragraph{Solve:}
Finally, we consider useful functions for solving systems of equations.
The first of these library functions is \textsc{solve}, which solves a given system of equations in terms of symbols,
which are imported or given as input.
Thus, \textsc{solve} yields symbolic solutions to this set of equation in terms of specified allowed symbols.
Of course these expressions can be complicated so rules are applied to simplify these expressions,
and this simplifying function is called \textsc{simplify}.
The role of \textsc{simplify} is to simplify algebraic expressions by recognizing and simplifying common or specified arithmetic expressions for symbols being both numbers and arithmetic operations,
for example the replacements $5+2$ by 7 or replacement of $1+2x+x^2-(1+x)^2$ by 0
with the end result of \textsc{simplify}
still being a symbol, albeit simpler in form.
Lastly,
inspired by a SymPy library function \textsc{free\_sym},
we introduce another equivalent function in our library,
which we call \textsc{symIdentifier};
this function returns all symbols found in expression entries of a given multidimensional array.

%*****************
%--------------------------------------------------
\subsection{New functions for the computer library}
\label{subsubsec:rlibrary}
%\label{subsubsec:alibrary}
In this appendix we discuss the computer library of our new functions which augments the set of known functions in public libraries discussed in~\S\ref{subsubsec:blibrary}. We briefly describe each new built-in function
that is used in our algorithms in \S\ref{subsec:pseudocoding}, and then we present the detailed information about these functions,
including types of input and output.
These functions are presented as their name in \textsc{camelCase} followed by a plain description of the function, then their inputs and outputs.

\begin{enumerate}
\item{Oracle for group/algebra properties:}
Our library includes one oracular function, 
accepting a binary input that determines the type of the group,
which is either a finite discrete group or a compact connected Lie group.
Next the input of this function is a name of a finite discrete group 
or a compact connected Lie group.
This oracular function also accepts an integer input corresponding to Hilbert-space dimension~$d$.
This function returns pertinent information about the properties of the input finite discrete group or the input compact connected Lie group, depending on the type of the group determined by the first input of the function.
This oracular function yields six outputs. 
Depending on the type of the group in the input, the first output is the number of inequivalent
irreps of the finite discrete group or the number of inequivalent irreps of Lie algebra corresponding to the compact connected Lie group with dimension less than or equal to~$d$.
The second output is the 
number of inequivalent representations of the finite discrete group or the number of inequivalent representations of the Lie algebra corresponding to the compact connected Lie group, with respect to the group type in the input,
for given~$d$.
The third output is the rank of the finite discrete group or the number of generators of the Lie algebra corresponding 
to the compact connected Lie group.
The fourth output is the dimension of the inequivalent irreps of the finite discrete group 
or the dimension of all inequivalent irreps of the compact connected Lie algebra corresponding to the compact connected Lie group, with dimension less than or equal to~$d$.
Next we have the fifth output,
which is all inequivalent unitary irreps of the finite discrete group or all inequivalent unitary irreps of the Lie algebra corresponding to the compact connected Lie group with dimension less than or equal to~$d$.
Finally,
the sixth output is all unitary representations of the finite discrete group or all unitary representations of the 
Lie algebra corresponding to the compact connected Lie group, for the given~$d$. 
\begin{itemize}
\item \textsc{props}
    $\triangleright$
        Properties and representations of a discrete group or a compact connected Lie group \textsc{gName}
    \begin{itemize}
        \item[] INPUT:
        \begin{itemize}
            \item [] \texttt{binary}[ ] \textsc{gType}
                $\triangleright$
                    Flag:
                    $0$ for finite discrete groups and $1$ for compact connected Lie groups.     
            \item [] \texttt{character}[ ] \textsc{gName}
                $\triangleright$
                    Name of finite discrete group or compact connected Lie group.
            \item [] \texttt{posinteger} \textsc{hDim}          $\triangleright$
                Hilbert-space dimension.
        \end{itemize}
        \item[] OUTPUT:
        \begin{itemize}
            \item [] \texttt{posinteger} \textsc{numIrrep}
                $\triangleright$ Number of inequivalent irreps of the \textsc{gName} for \textsc{gType}=0 and number of inequivalent irreps of algebra generating \textsc{gName} for \textsc{gType=1}
                with dimension less than or equal to \textsc{hDim}.
            \item [] \texttt{posinteger} \textsc{numRep}
                 $\triangleright$ Number of inequivalent reps of the \textsc{gName} for \textsc{gType}=0 and number of inequivalent irreps of algebra generating \textsc{gName} for \textsc{gType=1}
                with dimension \textsc{hDim}.
            \item [] \texttt{posinteger} \textsc{numGen}
                $\triangleright$ rank of \textsc{gName} for \textsc{gType}=0 and number of generators of the Lie algebra generating \textsc{gName} for \textsc{gType}=1.
            \item [] \texttt{posinteger}[ ] \textsc{dim} 
                $\triangleright$ Dimension of \textsc{numIrrep} inequivalent irreps of \textsc{gName} with dimension less than or equal to \textsc{hDim} for \textsc{gType=0} and dimension of \textsc{numIrrep} inequivalent irreps of Lie algebra generating \textsc{gName} with dimension less than or equal to \textsc{hDim}.
            \item [] \texttt{symjagged}[ ][ ][ ][ ]  \textsc{irrep} $\triangleright$ For \textsc{gType}=0
        entries of the first hyperrow label the \textsc{numIrrep} inequivalent irreps of \textsc{gName} with dimension less than or equal to \textsc{hDim}.
        Entries of the second hyperrow label \textsc{numGen} generators of \textsc{gName}.
        The third and fourth hyperrows are matrix elements for each irrep of \textsc{gName} with dimension less than or equal to \textsc{hDim}.  
        For \textsc{gType}=1 entries of the first hyperrow label the \textsc{numIrrep} inequivalent irreps of Lie algebra generating \textsc{gName} with dimension less than or equal to \textsc{hDim}.
        Entries of the second hyperrow label \textsc{numGen} generators of the Lie algebra generating \textsc{gName}.
        The third and fourth hyperrows are matrix elements for each irrep of Lie algebra generating \textsc{gName} with dimension less than or equal to \textsc{hDim}.  
         \item [] \texttt{symbol}[ ][ ][ ][ ]  \textsc{rep} $\triangleright$
        For $\textsc{gType}=0/1$ entries of the first hyperrow label the \textsc{numRep} inequivalent reps of \textsc{gName}/Lie algebra generating \textsc{gName} with dimension \textsc{hDim}.
        Entries of the second hyperrow label \textsc{numGen} generators of \textsc{gName}/Lie algebra generating \textsc{gName}.
        The third and fourth hyperrows are matrix elements for each rep of \textsc{gName}/Lie algebra generating \textsc{gName} with dimension \textsc{hDim}.   
        \end{itemize}
        \end{itemize}
\end{itemize}
%------------------
\item {}Imposing a trace preserving condition on completely positive maps: Our library requires a function for accepting $N$ number of $d\times d$ Kraus operators of a completely positive map and returning $N$ number of $d\times d$ Kraus operators of the same completely positive map that satisfy the trace preserving condition in Eq.~(\ref{eq:TP}). Given the set of Kraus operators, this function constructs  $\Xi$ in Eq.~(\ref{def:Xi}) and solves the matrix equation in Eq.~(\ref{eq:TP}) for parameters in Kraus operators of the completely positive map. If the solution exists, it employs the solution to simplify the Kraus operators and returns 1 for \textsc{TP}, indicating that the trace-preserving condition is satisfied in addition to the set of Kraus operators satisfying the trace-preserving condition~(\ref{eq:TP}). If the solution does not exist, the function returns 0 for \textsc{TP}, indicating that the solution does not exist. 
\begin{itemize}
    \item\textsc{solveChannel}
    $\triangleright$
    Apply the trace-preserving condition to the input CP map to yield either a channel or no solution.
 \begin{itemize}
      \item[] INPUT:
      \begin{itemize}
          \item[] 
     \texttt{symbol}[\textsc{n}][\textsc{d}][\textsc{d}]~\textsc{kraus}
      $\triangleright$ \textsc{n} 
      $\textsc{d}\times \textsc{d}$
      Kraus matrices for CP map.
       \end{itemize}
      \item[] OUTPUT:
      \begin{itemize}
          \item [] \texttt{binary}[ ] \textsc{TP}
                $\triangleright$
                    Flag:
                    $0$ for no solution and $1$ for existence of the solution.
         \item[]\texttt{symbol}[\textsc{n}][\textsc{d}][\textsc{d}]~\textsc{kraus}
      $\triangleright$ \textsc{n} 
      $\textsc{d}\times \textsc{d}$
      Kraus matrices for channel.
       \end{itemize}
      \item[] PROCEDURE:
        \begin{enumerate}
            \item \texttt{symbol}[\textsc{d}][\textsc{d}]~\textsc{xi}
            $\triangleright$ Represents $\Xi$~(\ref{def:Xi}).
            \item \texttt{symbol}[$\frac{\textsc{d}(\textsc{d}-1)}{2}$] \textsc{eqs}
            $\triangleright$ Constrained expressions due to the trace-preserving condition.
            \item \texttt{symbol}[ ] \textsc{symbols}
                $\triangleright$
                Yields parameters obtained from \textsc{kraus}. 
            \item \texttt{symbol}[ ] \textsc{sol}
            $\triangleright$ Relations between parameters in \textsc{kraus} due to the trace-preserving constraint;
            same size as~\textsc{symbols}.
            \item \texttt{binary} \textsc{flag}
                $\triangleright$
                    TRUE if solutions exist.
            \item Compute \textsc{xi}
                $\triangleright$ From Eq.~(\ref{def:Xi})
           \item For all $\textsc{i}\in[\textsc{d}]$,  \textsc{eqs}[\textsc{i}]$\gets\textsc{xi}[\textsc{i}][\textsc{i}]-1$
           $\triangleright$ Apply trace-preserving condition for diagonal elements.
            \item For \textsc{i} from $\textsc{d}+1$ to $\frac{\textsc{d}(\textsc{d}-1)}{2}$, \textsc{eqs} gets upper-diagonal elements of \textsc{xi}
            $\triangleright$ Apply
                trace-preserving condition for upper-diagonal elements.
            \item \textsc{symbols}$\gets$\textsc{symIdentifier}(\textsc{kraus})
                $\triangleright$ Extracts symbols from~\textsc{kraus}.
                \item \textsc{sol}$\gets$\textsc{solve}(\textsc{eqs},\textsc{params}) $\triangleright$
            Assigns \textsc{eqs}=0,
            then solves for \textsc{params};
            \textsc{flag}$\gets$FALSE if a solution is not found.
            \item If \textsc{flag}, proceed to the next step.
            \item RETURN \textsc{kraus}$\gets$ \textsc{simplify}(\textsc{kraus},\textsc{sol})
                $\triangleright$
                    Employ expressions in \textsc{sol} to simplify \textsc{kraus}
        \end{enumerate}
 \end{itemize}
\end{itemize}
%------------------
\item{Solving systems of linear equations:}
The next function solves systems of homogeneous linear equations symbolically;
this system is $A_n\bm{x}=\bm0$ for $(d_1\times d_2)$-dimensional symbolic matrices $A_n$s 
with $n\in[N]$.
Solving this system of equations is accomplished by computing the intersection between kernels of $\left\{A_n\right\}$,
which is accomplished by solving a system of linear homogeneous equations.
The matrix of coefficients in this system of coupled linear equations
is an $Nd_1\times d_2$ matrix constructed by stacking $N$ instances of $A_n$ matrices.
This system of homogeneous linear equations is solved by standard methods such as, singular value decomposition.
\begin{itemize}
\item\textsc{symSolve}
    $\triangleright$ 
        Solves symbolic~$\bm{x}$
        for $A_n\bm{x}=\bm0\forall n\in[N]$,
        with each $A_n$ a $\textsc{k}\times \textsc{l}$ symbolic matrix,
        and \textsc{simplify}
        employed to simplify all algebraic expressions.
  \begin{itemize}
      \item[] INPUT: \texttt{symbol}[\textsc{n}][\textsc{k}][\textsc{l}] \textsc{a}
      \item[] OUTPUT: \texttt{symbol}[\textsc{l}] \textsc{x}
 \end{itemize}
\end{itemize}
   
\item{Reshape:}
Our library incorporates a function that reshapes the input symbolic vector into a square matrix,
where the input vector has a length
that is a squared integer~$d^2$ and the matrix has size $d\times d$.
Reshaping is accomplished by writing in order each element of the vector into each element of the first row of the matrix until that row is full.
Then we continue by writing the next elements of the vector into the next row of the matrix until that row is filled.
This procedure is complete when all rows of the vector are written into all elements of the square matrix,
and we have ensured that our matrix is exactly the right size for this transcription from a length~$d^2$ vector to work properly.
\begin{itemize}
    \item 
\textsc{reshape}
    $\triangleright$
        Reshape symbolic \textsc{vector}~\S\ref{subsubsec:blibrary} 
            but here only for a vector of squared-integer length to a square matrix,
            converted according to the rule that the first row fills the matrix, then the second row and so on.
\begin{itemize}
    \item [] INPUT:
        \texttt{symbol}$\left[\textsc{Dim}^2\right]$ \textsc{vector}
    \item [] OUTPUT:
        \texttt{symbol}[\textsc{Dim}][\textsc{Dim}] \textsc{matrix}
\end{itemize}
\end{itemize}
\item {Convert a vector to multiple matrices:}
Our library furthermore requires a function 
that converts a given vector to a set of square matrices,
which generalizes the previous function
that maps a vector to a single square matrix.
The input is a symbolic vector of size $Kd^2\times 1$ yielding~$K$ number of $d\times d$ matrices at the output. 
This function chops an input vector of size $Kd^2$ into~$K$  length~$d^2$ vectors and then reshapes each of these vectors into a square $d\times d$ matrix. 
\begin{itemize}
    \item 
\textsc{vecToMatr}
       $\triangleright$ Chops a length-$KD^2$  symbolic vector into~$K$ length~$D^2$ vectors and \textsc{reshape} each vector to a square~$D$-dimensional symbolic.
       \begin{itemize}
           \item [] INPUT: 
           \begin{itemize}
           \item [ ]\texttt{symbol}$\left[\textsc{k}*\textsc{d}^2\right]$ \textsc{vector}, \item[ ]\texttt{posinteger} \textsc{k}
           \end{itemize}
           \item [] OUTPUT: \texttt{symbol}[\textsc{k}][\textsc{d}][\textsc{d}] \textsc{matrixSeq}
       \end{itemize}
       \end{itemize}
\end{enumerate}

%-------------------------------------------------
\bibliography{qchannels}
\end{document}